\documentclass[10pt,journal]{IEEEtran} 
\IEEEoverridecommandlockouts
\usepackage[letterpaper,margin=1in]{geometry}

\usepackage{cite}
\usepackage{amsmath,amssymb,amsfonts}
\usepackage{empheq}
\usepackage[bb=dsserif]{mathalpha}
\usepackage{algorithmic}
\usepackage{graphicx}
\usepackage{textcomp}
\usepackage{xcolor}
\usepackage{tikz,hyperref}
\usepackage{multirow,bigstrut}

\def\BibTeX{{\rm B\kern-.05em{\sc i\kern-.025em b}\kern-.08em
    T\kern-.1667em\lower.7ex\hbox{E}\kern-.125emX}}

\newcommand{\Enc}{\mathsf{Enc}}
\newcommand{\Dec}{\mathsf{Dec}}

\newcommand{\Eve}{\mathsf{Eve}}

\newtheorem{theorem}{Theorem}

\newtheorem{remark}{Remark}
\newtheorem{definition}{Definition}
\newtheorem{proposition}{Proposition}

\newtheorem{corollary}{Corollary}

\begin{document}

\title{Secure Rate-Distortion-Perception: A Randomized Distributed Function Computation Approach for Realism}

\author{
  \IEEEauthorblockN{Gustaf Åhlgren\textsuperscript{1} and Onur G\"unl\"u\textsuperscript{2, 1}}\\
    \IEEEauthorblockA{\textsuperscript{1}%
        Information Theory and Security Laboratory (ITSL), Link{\"o}ping University, Sweden
    }\\
    \IEEEauthorblockA{\textsuperscript{2}%
        Lehrstuhl für Nachrichtentechnik, Technische Universität Dortmund, Germany\\
        gustaf.ahlgren@liu.se, onur.guenlue@tu-dortmund.de
    }

}

\maketitle

\begin{abstract}
Fundamental rate–distortion–perception (RDP) trade-offs arise in applications requiring maintained perceptual quality of reconstructed data, such as neural image compression. When compressed data is transmitted over public communication channels, security risks emerge. We therefore study secure RDP under negligible information leakage over both noiseless channels and broadcast channels (BCs) with correlated noise components. For noiseless channels, the exact secure RDP region is characterized. For BCs, an inner bound is derived and shown to be tight for a class of more-capable BCs. Separate source–channel coding is further shown to be optimal for this exact secure RDP region with unlimited common randomness available. Moreover, when both encoder and decoder have access to side information correlated with the source and the channel is noiseless, the exact RDP region is established. If only the decoder has correlated side information in the noiseless setting, an inner bound is derived along with a special case where the region is exact. Binary and Gaussian examples demonstrate that common randomness can significantly reduce the communication rate in secure RDP settings, unlike in standard rate–distortion settings. Thus, our results illustrate that random binning-based coding achieves strong secrecy, low distortion, and high perceptual quality simultaneously.
\end{abstract} 


%
%
\section{Introduction}


Modern communication networks are increasingly designed to convey the information that is most relevant to the end application, rather than to reproduce signals with exact bit‑level fidelity. This transition has motivated substantial interest in semantic communication techniques \cite{GunduzBITS, GunduzSemantic}, which prioritize the transmission of the most informative feature-domain representations of the source data. Such approaches are advantageous in scenarios with limited bandwidth or latency limitations, including autonomous driving and immersive multimedia systems such as augmented or virtual reality \cite{GunduzBITS}.

From an information‑theoretic viewpoint, semantic communication can be formulated as a remote source coding problem \cite{DobrushinRemote, BergerBook,Csiszarbook}, where the encoder observes a source indirectly and the decoder aims to recover the source and not necessarily the encoder's noisy observation of that source. Our recent work has introduced the randomized distributed function computation (RDFC) framework \cite{MyWIFS}, which incorporates a controlled randomization mechanism at the encoder. This mechanism enables the decoder to generate outputs that follow a specified target distribution while ensuring that essential system functionality is retained. The RDFC framework provides strong performance guarantees by utilizing results from strong coordination \cite{CuffChannelSynthesis}, and these guarantees remain valid even when the available common randomness is limited or entirely absent \cite{MyWIFS,OnurEUCNC,OurDeepRDFC,GustafISIT2025}. The RDFC framework has been applied in areas such as neural compression using generative models \cite{Flamich2020,HavasiMRC}, federated learning settings with side information \cite{Isik2024adaptive}, and mechanisms designed to provide differential privacy \cite{MyWIFS,CheukDP,OnurRDFCjournal,KairouzLDP}.

In this work, we focus on the rate–distortion–perception (RDP) problem \cite[Section~17.4.2]{BookforDataCompression}, \cite{RethinkingRDP,Matsumoto2018RDP,Chen2022RDP,KhistiUniversalRDP,TheisWagner2021codingforDP,YukselCR}, which is closely tied to image compression \cite{JonaE2E} and can be seen as a special case of the RDFC framework due to its randomized nature. The RDP setting seeks to jointly reduce the expected distortion between the original image and its reconstruction while also constraining the reconstruction to follow a distribution that resembles that of the source \cite{Wagner2022rate,DenizRDP}. The latter requirement enforces high perceptual quality in the reconstruction, such that the generated outputs exhibit statistical characteristics aligned with the original image. Such a constraint is also captured in generative discriminators used in deep learning–based models \cite{Goodfellow2014generative,Wasserstein,ImprovedWasserstein}.


Security and privacy play a critical role in practical systems, in which the encoder outputs may be exposed to unauthorized observers in the network. The security risk is amplified for joint source-channel coding schemes \cite{SecureDeepJSCC,DiffSecureDeepJSCC}. Physical-layer security methods offer a means to counter such threats \cite{OurSecureFunction, CuffRDSecrecy,Onur2026lowRDP}. In this work, we consider secure RDP under strong secrecy, for which the amount of leakage is negligible, in two fundamental transmission scenarios. The first scenario is a noiseless communication model, which can correspond to end-to-end image compression at higher open systems interconnection (OSI) layers. The second scenario involves a noisy broadcast channel (BC), where correlated noise components across multiple receivers introduce additional complexity in determining the achievable rate, distortion, and perceptual constraints. For the noiseless case, we characterize the exact secure RDP region. For the BC setting, we develop an inner bound on the secure RDP region by employing a refined random-binning method to ensure strong secrecy. Furthermore, under a more-capable channel assumption, we obtain an exact description of the secure RDP region for transmissions over BCs.


The source-channel separation theorem \cite{shannon1948OGPaper} motivated the layered system designs used in communication networks. Similarly, recent work has examined such separation results for RDP when the encoder and decoder aim to satisfy a perception constraint, which is imposed on the reconstructed image as a probability distribution constraint. The source-channel separation result in \cite{QuAlmostRDPdifferentDist} examines the optimality of separation when the encoder and decoder aim to produce an independent and identically distributed (i.i.d.) target probability distribution at the decoder, which is similar to the perception (or realism) constraint considered in this paper. Similarly, \cite{tian2025source} analyzes the separation for the RDP problem for more general realism constraints and shows that separation is not optimal for block-level realism constraints if common randomness is not available at the encoder and the decoder. Motivated by these results, we analyze source-channel separation for a special case of secure RDP over BCs when an unlimited amount of common randomness is available.


For learned image compression systems, the previously processed images could also provide correlated side information to the encoder and/or decoder. Thus, as extensions of the classical distributed lossy source coding with side information results in \cite{WynerZiv} that examine how side information influences the rate-distortion performance, recent work considers the RDP framework with side information \cite{WagnerHamidGunduzarXiv2025,HamidGunduzISIT2023,NiuCondRDP}. In this work, we consider two formulations of side information, depending on its availability, when we assume that the channel is noiseless. In the first setting, both the encoder and decoder have access to the side information that is correlated with the input data. In the second setting, the side information is available only at the decoder. For the former setting, we characterize the corresponding secure RDP region. For the latter setting, we establish an inner bound for the secure RDP setting and show that this bound is exact for a special case.

\subsection{Main Contributions}
The main contributions of this work include:
\begin{itemize}
    \item We characterize the exact secure RDP region for noiseless communication channels under a strong secrecy constraint on the reconstructed sequence;
    \item We derive an inner bound on the secure RDP region for noisy memoryless BCs and show that this inner bound is tight for some more-capable BCs to establish an exact rate region characterization for that special case;
    \item We establish the exact secure RDP region when side information is available at both the encoder and decoder and the communication channel is noiseless;
    \item For noiseless channels, we obtain a general inner bound on the secure RDP region when the side information is available only at the decoder and identify a special case for which the secure RDP region is exact; and
    \item We provide a binary example illustrating how common randomness can substantially reduce the communication rate in secure RDP settings, which is a gain not attainable in standard rate-distortion formulations.
\end{itemize}

These results follow mainly from our previous conference papers \cite{GustafISIT2025,GustafWCNC2026}. Moreover, the new contributions in this work include:
\begin{itemize}
    \item We characterize the optimality of source-channel coding separation over a noisy memoryless BC such that the legitimate receiver has a more-capable BC than the eavesdropper when there is unlimited common randomness.
    \item Finally, we evaluate a Gaussian source example to illustrate how decoder side information and common randomness shape the secure RDP trade-off.
\end{itemize}

%
%
\subsection{Paper Outline}
Section~\ref{sec:sysmodel} provides the secure RDP system models and the main definitions. Section~\ref{sec:ISIT_results} establishes the rate regions for secure RDP over noiseless and noisy channels. In Section~\ref{sec:WCNC_results}, we extend the secure RDP over noiseless channel results to two different side information scenarios. In Section~\ref{sec:SCS_analysis}, we analyze the optimality of source-channel separation for a special case of secure RDP over BCs with unlimited common randomness. 
Section~\ref{sec:examples} evaluates some of the secure RDP regions for relevant examples, including Gaussian sources. 

%
%
\subsection{Notation}
Let $X$ denote a random variable with the alphabet $\mathcal{X}$, and $x$ denotes its realization. Abbreviate $P_{X}(x)$ as $P_{X}$, denoting the probability distribution of a random variable $X$. Denote the set of integers $\{a, a+1,\ldots,b\}$ as $[a:b]$, and the set of indices $\{1,2, \ldots, i\!-\!1,i\!+\!1,\ldots,n\}$ as $\{n\} \backslash i$, respectively. $X \sim \text{Unif}[a:b]$ represents a random variable $X$ that is uniformly distributed on the interval between $a$ and $b$. The total variation (TV) distance between two probability distributions is denoted
\begin{equation}
    ||P_{X} - P_{Y}||_{\text{TV}} \triangleq \frac{1}{2}\sum_{a \in \mathcal{A}} |P_{X}(a) - P_{Y}(a)|.
\end{equation}
Let $\mathbb{1}\{\cdot\}$ denote the indicator function. For $a, b \in [0,1]$, we define the *-operator as $a*b = a(1-b) + b(1-a)$.
Let $P_{X} \approx P_{Y}$ denote that the two probability distributions $P_{X}$ and $P_{Y}$ are close in TV distance. Let $H_{b}(p) = - p \log{p} - (1 - p) \log{(1-p)}$ denote the binary entropy function,
with $0 \cdot \log{0} \triangleq 0$. A binary symmetric channel (BSC) with cross-over probability $\alpha\in[0,1]$ is denoted as BSC$(\alpha)$. A Bernoulli random variable with success probability $\beta\in[0,1]$ is denoted as \text{Bern}$(\beta)$.

%
%
\section{System Models and Definitions}\label{sec:sysmodel}
For all RDP system models analyzed in this paper, we consider an encoder that observes a source sequence $X^{n} \sim Q_{X}^{n}$ taking values in the finite Polish alphabet $\mathcal{X}^{n}$ along with  common randomness $C \in [1: 2^{nR_{0}}]$, the latter of which is also available to the decoder. The encoder then outputs either a channel input $\widetilde{X}^{n}$ or an index $S \in [1:2^{nR}]$ that is communicated over either a noisy memoryless BC or a noiseless channel, respectively. The encoder output $\widetilde{Y}^{n}$ (which is a noisy version of $\widetilde{X}^{n}$) or $S$ is received by a decoder. Moreover, an eavesdropper (Eve) in the network observes another channel output sequence $\widetilde{Z}^{n}$ or the index $S$. Observing the common randomness $C$ and either of the noisy or noiseless channel outputs ($\widetilde{Y}^{n}$ or $S$), the decoder outputs $Y^{n}\in\mathcal{X}^{n}$ such that (i) the expected distortion between $X^{n}$ and $Y^{n}$ is minimized, (ii) the induced distribution of the output sequence $P_{Y^{n}}$ is close in TV distance to the input distribution $Q_{X}^{n}$, and (iii) strong secrecy is satisfied, which is imposed between $Y^n$ and either of the noisy or noiseless channel outputs. We remark that such a secrecy constraint is relevant in the context of generative artificial intelligence where the output image is the commodity to keep confidential, as discussed in \cite{GustafISIT2025,GustafWCNC2026}. For our distortion metric we consider a distortion $d(x^{n},y^{n})$ such that
\begin{equation}
    d(x^{n},y^{n}) = \frac{1}{n} \sum^{n}_{i=1}d(x_{i},y_{i}),
\end{equation}
with $d:\mathcal{X}\times\mathcal{Y} \rightarrow [0,\infty)$. Fig.~\ref{fig:simplified_SRDP} depicts the system model for the secure RDP over a noiseless channel, whose rate region definition is given below.

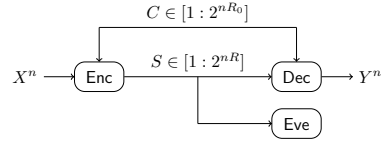
\begin{figure}[t]
    \centering
    \resizebox{0.65\linewidth}{!}{
    \begin{tikzpicture}
		\node (in) at (-2.5,0) {$X^{n}$};
        \node (enc) at (-1,0) [draw,rounded corners = 6pt, minimum width=1cm,minimum height=0.6cm, align=left] {$\Enc$};
        \node (dec) at (3,0) [draw,rounded corners = 6pt, minimum width=1cm,minimum height=0.6cm, align=left] {$\Dec$};
        \node (eve) at (3,-0.95) [draw,rounded corners = 6pt, minimum width=1cm,minimum height=0.6cm, align=left] {$\Eve$};
        \node (out) at (4.5,0) {$Y^{n}$};
        \node (common) at (1,1.3) {$C \in [1:2^{nR_{0}}]$};
        \node (message) at (1,.3) {$S \in [1:2^{nR}]$};
        \draw[<-] (enc) -- (-1,1);
        \draw (-1,1) -- (3,1);
        \draw[<-] (dec) -- (3,1);
        \draw[->] (in) -- (enc);
        \draw[->] (dec) -- (out);
        \draw[->] (enc) -- (dec);
        \draw (1,0) -- (1,-0.95);
        \draw[->] (1,-0.95) -- (2.5,-0.95);
    \end{tikzpicture}
    }
    \caption{The system model of the secure RDP problem with transmissions over a noiseless channel.}
    \label{fig:simplified_SRDP}
\end{figure}

\begin{definition}\label{def:SecureRDPNoiseless}
    A secure RDP tuple $(R,R_{0},D)$ is \emph{achievable} for an i.i.d target output distribution $Q_{X}^{n}$ if, for any $\epsilon\!>\!0$, there exist $n\!\geq\!1$, an encoder, and a decoder such that
    \begin{align}
        &||P_{Y^{n}} - Q_{X}^{n}||_{TV} \leq \epsilon \label{eq:real_SRDP} && \text{(realism)} \\
        &\mathbb{E}[d(X^{n},Y^{n})] \leq D + \epsilon \label{eq:dist_SRDP} && \text{(distortion)} \\
        &I(Y^{n};S) \leq \epsilon \label{eq:secrecy_SRDP} && \text{(strong secrecy)}.
    \end{align}
    The \emph{secure} RDP region $\mathcal{R}$ is the closure of the set of all achievable tuples.\hfill $\lozenge$
\end{definition}

Fig.~\ref{fig:noisy_SRDP_broadcast} depicts the system model for secure RDP over a noisy memoryless BC, whose rate region definition is given below.

\begin{figure}[t]
    \centering
    \resizebox{0.65\linewidth}{!}{
    \begin{tikzpicture}
		\node (in) at (-2.5,0) {$X^{n}$};
        \node (enc) at (-1,0) [draw,rounded corners = 6pt, minimum width=1cm,minimum height=0.6cm, align=left] {$\Enc$};
        \node (ch_x) at (0,0.3) {$\widetilde{X}^n$};
        \node (ch_in) at (0.88,0) {};
        \node (channel) at (1.5,-0.45) [draw,rounded corners = 6pt, minimum width=1.5cm,minimum height=1.4cm, align=left] {$P_{\widetilde{Y}\widetilde{Z}|\widetilde{X}}$};
        \node (ch_out_y) at (2.12,0) {};
        \node (ch_out_z) at (2.12,-0.95) {};
        \node (ch_y) at (2.9,0.3) {$\widetilde{Y}^n$};
        \node (dec) at (4,0) [draw,rounded corners = 6pt, minimum width=1cm,minimum height=0.6cm, align=left] {$\Dec$};
        \node (out) at (5.5,0) {$Y^{n}$};
        \node (eve) at (4,-0.95) [draw,rounded corners = 6pt, minimum width=1cm,minimum height=0.6cm, align=left] {$\Eve$};
        \node (eve_out) at (2.9,-0.65) {$\widetilde{Z}^n$};
        \node (common) at (1.5,1.3) {$C \in [1:2^{nR_{0}}]$};
        \draw[->] (in) -- (enc);
        \draw[->] (enc) -- (ch_in);
        \draw[->] (ch_out_z) -- (3.5,-0.95);
        \draw[->] (dec) -- (out);
        \draw[->] (ch_out_y) -- (dec);
        \draw[<-] (enc) -- (-1,1);
        \draw (-1,1) -- (4,1);
        \draw[<-] (dec) -- (4,1);
    \end{tikzpicture}
    }
    \caption{The system model of the secure noisy RDP problem with transmissions over a noisy memoryless BC.}
    \label{fig:noisy_SRDP_broadcast}
\end{figure}
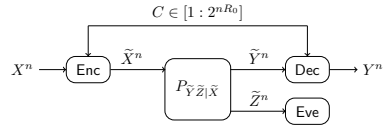
\begin{definition}\label{def:SecureRDPNoisy}
    A secure noisy RDP tuple $(R, R_{0}, D)$ is achievable given a noisy memoryless BC $P_{\widetilde{Y}\widetilde{Z}|\widetilde{X}}$, and an i.i.d target output distribution $Q_{X}^{n}$ if for $n \geq 1$ there exists an encoder and decoder that satisfy the realism and distortion constraints in (\ref{eq:real_SRDP}) and (\ref{eq:dist_SRDP}), and
    \begin{align}
        &I(Y^{n};\widetilde{Z}^{n}) \leq \epsilon \label{eq:secrecy_noisyRDP} && \text{(strong secrecy)}.
    \end{align}
    The \emph{secure} noisy RDP region $\mathcal{R}_{N}$ is the closure of the set of all achievable tuples.\hfill $\lozenge$
\end{definition}

We next provide a definition of a more-capable BC, for which we establish the exact secure noisy RDP region in Section~\ref{sec:ISIT_results}.

\begin{definition}\label{def:morecapabledef}
    A memoryless BC $P_{\widetilde{Y}\widetilde{Z}|\widetilde{X}}$ is more-capable if, for all $P_{\widetilde{X}}$, we have 
    \begin{align}
        I(\widetilde{Y};\widetilde{X})\geq I(\widetilde{Z};\widetilde{X}) \label{eq:morecapabledef}.
    \end{align}
\end{definition}

%
%

In Section~\ref{sec:WCNC_results}, we also consider the secure RDP problem when the decoder has access to side information $Z^{n} \in \mathcal{Z}^{n}$, where $\mathcal{Z}$ is a finite Polish space, that is correlated with the source sequence $X^{n}$ according to $P_{X^{n}Z^{n}} = Q_{XZ}^{n}$. Similarly, the encoder observes the source sequence $X^{n}$, the common randomness $C$, and the side information $Z^{n}$ which might be available also to the encoder. In both systems considered in Section~\ref{sec:WCNC_results}, the encoder produces an index $S$ that is observed by the decoder and the eavesdropper and the decoder produces the sequence $Y^{n}$ from $S$, $C$, and $Z^{n}$ to (i) reduce the expected distortion between $X^{n}$ and $Y^{n}$, (ii) ensure that the output probability distribution of $Y^{n}$ is close to that of the input sequence $X^{n}$ in TV, and (iii) achieve strong secrecy imposed between $S$ and $Y^{n}$. In the system depicted in Fig.~\ref{fig:SecRDPwithSIatBoth}, both the encoder and the decoder have access to the side information, and in Fig.~\ref{fig:SecRDPwithSIatDec}, only the decoder has access to the side information $Z^{n}$.

We next give a definition of the rate region for secure RDP with side information.

\begin{figure}[t]
    \centering
    \resizebox{0.65\linewidth}{!}{
    \begin{tikzpicture}
		\node (in) at (-2.5,0) {$X^{n}$};
        \node (si) at (-2.5,-1.5) {$Z^{n}$};
        \node (enc) at (-1,0) [draw,rounded corners = 6pt, minimum width=1cm,minimum height=0.6cm, align=left] {$\Enc$};
        \node (dec) at (3,0) [draw,rounded corners = 6pt, minimum width=1cm,minimum height=0.6cm, align=left] {$\Dec$};
        \node (eve) at (2,-0.70) [draw,rounded corners = 6pt, minimum width=1cm,minimum height=0.6cm, align=left] {$\Eve$};
        \node (out) at (4.5,0) {$Y^{n}$};
        \node (common) at (1,1.3) {$C \in [1:2^{nR_{0}}]$};
        \node (message) at (1,.3) {$S \in [1:2^{nR}]$};
        \draw[<-] (enc) -- (-1,1);
        \draw (-1,1) -- (3,1);
        \draw[<-] (dec) -- (3,1);
        \draw[->] (in) -- (enc);
        \draw[->] (dec) -- (out);
        \draw[->] (enc) -- (dec);
        \draw (1,0) --  (1,-0.7);
        \draw[->] (1,-0.7) -- (eve);
        \draw (si) -- (3,-1.5);
        \draw[->] (-1,-1.5) -- (enc);
        \draw[->] (3,-1.5) -- (3,-0.3);
    \end{tikzpicture}
    }
    \caption{The system model of the secure RDP problem where the side information is available at both the encoder and the decoder.}
    \label{fig:SecRDPwithSIatBoth}
\end{figure}
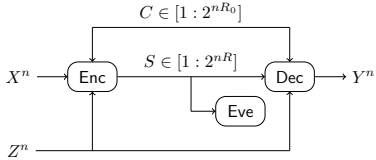

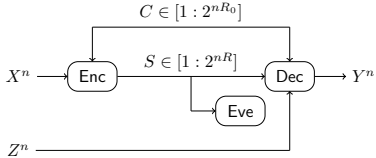
\begin{figure}[t]
    \centering
    \resizebox{0.65\linewidth}{!}{
    \begin{tikzpicture}
		\node (in) at (-2.5,0) {$X^{n}$};
        \node (si) at (-2.5,-1.5) {$Z^{n}$};
        \node (enc) at (-1,0) [draw,rounded corners = 6pt, minimum width=1cm,minimum height=0.6cm, align=left] {$\Enc$};
        \node (dec) at (3,0) [draw,rounded corners = 6pt, minimum width=1cm,minimum height=0.6cm, align=left] {$\Dec$};
        \node (eve) at (2,-0.70) [draw,rounded corners = 6pt, minimum width=1cm,minimum height=0.6cm, align=left] {$\Eve$};
        \node (out) at (4.5,0) {$Y^{n}$};
        \node (common) at (1,1.3) {$C \in [1:2^{nR_{0}}]$};
        \node (message) at (1,.3) {$S \in [1:2^{nR}]$};
        \draw[<-] (enc) -- (-1,1);
        \draw (-1,1) -- (3,1);
        \draw[<-] (dec) -- (3,1);
        \draw[->] (in) -- (enc);
        \draw[->] (dec) -- (out);
        \draw[->] (enc) -- (dec);
        \draw (1,0) --  (1,-0.7);
        \draw[->] (1,-0.7) -- (eve);
        \draw (si) -- (3,-1.5);
        \draw[->] (3,-1.5) -- (3,-0.3);
    \end{tikzpicture}
    }
    \caption{The system model of the secure RDP problem where the side information is available at the decoder only.}
    \label{fig:SecRDPwithSIatDec}
\end{figure}

\begin{definition}\label{def:secRDPwithSI}
    A \emph{secure} RDP tuple $(R, R_{0}, D)$ is achievable with side information $Z^{n} \in \mathcal{Z}^{n}$, correlated with the source sequence $X^{n}$ according to $Q_{XZ}^{n}$, if, for any $\epsilon > 0$ there exists $n \geq 1$, an encoder, and a decoder that satisfy the realism and distortion constraints in (\ref{eq:real_SRDP}) and (\ref{eq:dist_SRDP}), and
    \begin{align}
        &I(Y^{n};S) \leq \epsilon \label{eq:secrecy_SRDPwithSI} && \text{(strong secrecy)}.
    \end{align}
    The rate region for the secure RDP with side information is the closure of the set of all achievable tuples. For secure RDP with side information available at both the encoder and the decoder, the rate region is denoted as $\mathcal{R}_{\text{SI,ED}}$, and for secure RDP with side information available only at the decoder as $\mathcal{R}_{\text{SI,D}}$. \hfill $\lozenge$
\end{definition}

We next define a special case where the side information $Z^{n}$ and the output $Y^{n}$ are pairwise i.i.d, for which we establish the exact $\mathcal{R}_{\text{SI,D}}$ in Section~\ref{sec:WCNC_results}.

\begin{definition}\label{def:jointlyiid}
    Two sequences of random variables $A^{n}$ and $B^{n}$ are jointly i.i.d if and only if we have
    \begin{equation}
        P_{A^{n}B^{n}} = P_{AB}^{n}.
    \end{equation} \hfill $\lozenge$
\end{definition}

In Section~\ref{sec:SCS_analysis}, we analyze the optimality of separate source and channel coding for the secure RDP problem with unlimited common randomness, for which we can achieve the minimal index rate $R$. The considered system model is depicted in Fig.~\ref{fig:noisy_SRDP_SCS}. Different from the system model in Fig.~\ref{fig:noisy_SRDP_broadcast}, the number of channel uses is $m\geq 1$, which is not necessarily equal to the sequence length $n$. Thus, the secrecy constraint here differs from (\ref{eq:secrecy_noisyRDP}) because the channel sequences and the source and reconstruction sequences have different lengths. Similar to \cite{tian2025source}, the definition of the achievable region for the secure RDP problem with unlimited common randomness and transmission over a noisy memoryless BC is given in Definition~\ref{def:SecureRDPNoisy_SCS} below.

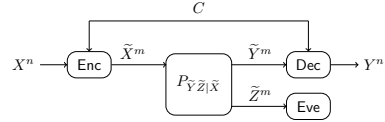
\begin{figure}[t]
    \centering
    \resizebox{0.65\linewidth}{!}{
    \begin{tikzpicture}
		\node (in) at (-2.5,0) {$X^{n}$};
        \node (enc) at (-1,0) [draw,rounded corners = 6pt, minimum width=1cm,minimum height=0.6cm, align=left] {$\Enc$};
        \node (ch_x) at (0,0.3) {$\widetilde{X}^m$};
        \node (ch_in) at (0.88,0) {};
        \node (channel) at (1.5,-0.45) [draw,rounded corners = 6pt, minimum width=1.5cm,minimum height=1.4cm, align=left] {$P_{\widetilde{Y}\widetilde{Z}|\widetilde{X}}$};
        \node (ch_out_y) at (2.12,0) {};
        \node (ch_out_z) at (2.12,-0.95) {};
        \node (ch_y) at (2.9,0.3) {$\widetilde{Y}^m$};
        \node (dec) at (4,0) [draw,rounded corners = 6pt, minimum width=1cm,minimum height=0.6cm, align=left] {$\Dec$};
        \node (out) at (5.5,0) {$Y^{n}$};
        \node (eve) at (4,-0.95) [draw,rounded corners = 6pt, minimum width=1cm,minimum height=0.6cm, align=left] {$\Eve$};
        \node (eve_out) at (2.9,-0.65) {$\widetilde{Z}^m$};
        \node (common) at (1.5,1.3) {$C$};
        \draw[->] (in) -- (enc);
        \draw[->] (enc) -- (ch_in);
        \draw[->] (ch_out_z) -- (3.5,-0.95);
        \draw[->] (dec) -- (out);
        \draw[->] (ch_out_y) -- (dec);
        \draw[<-] (enc) -- (-1,1);
        \draw (-1,1) -- (4,1);
        \draw[<-] (dec) -- (4,1);
    \end{tikzpicture}
    }
    \caption{The system model of the secure noisy RDP problem with transmissions over a noisy memoryless BC when there is enough common randomness $C$ available at the encoder and the decoder.}
    \label{fig:noisy_SRDP_SCS}
\end{figure}

\begin{definition}\label{def:SecureRDPNoisy_SCS}
    A secure noisy RDP tuple $(R, D)$ is achievable given a noisy memoryless BC $P_{\widetilde{Y}\widetilde{Z}|\widetilde{X}}$ and with unlimited common randomness if for $n,m \geq 1$ there exists an encoder and decoder that satisfy the realism and distortion constraints in (\ref{eq:real_SRDP}) and (\ref{eq:dist_SRDP}), and
    \begin{align}
        &\frac{m}{n} \leq \kappa + \epsilon && \text{(mismatch factor)} \\ 
        &I(Y^{n};\widetilde{Z}^{m}) \leq \epsilon \label{eq:secrecy_noisyRDP_SCS} && \text{(strong secrecy)}.
    \end{align}
    where $\kappa$ is the channel mismatch factor for the block-lengths $n$ and $m$ for the source and the channel sequences. The \emph{secure} noisy RDP region with unlimited common randomness $\mathcal{R}_{\text{N,UC}}$ is the closure of the set of all achievable tuples.\hfill $\lozenge$
\end{definition}

The achievable region with joint source-channel coding for the secure noisy RDP problem with unlimited common randomness is denoted $\mathcal{R}_{\text{N,UC,J}}$, and the achievable region with separate source channel coding is denoted $\mathcal{R}_{\text{N,UC,S}}$.




The main rate-region definitions are outlined in Table~\ref{tab:rateRegions} for convenience. 

{\renewcommand{\arraystretch}{1.3}
\linespread{1.7} 
\begin{table*}[t]
\caption{An overview of the main rate regions defined.}
\label{tab:rateRegions}
    \begin{center}
        \begin{tabular}{|c|c|c|c|c|}
        \hline
        Region & Channel & Side information & Common randomness & Secrecy constraint \\
        \hline
        \hline
        $\mathcal{R}$ & Noiseless & - & at rate $R_{0}$ & $I(Y^{n};S) < \epsilon$  \\
        \hline
        $\mathcal{R}_{\text{N}}$ & Noisy BC  & - & at rate $R_{0}$ & $I(Y^{n};\widetilde{Z}^{n}) < \epsilon$ \\
        \hline
        $\mathcal{R}_{\text{SI,ED}}$ & Noiseless & at encoder \& decoder & at rate $R_{0}$ & $I(Y^{n};S) < \epsilon$ \\ 
        \hline
        $\mathcal{R}_{\text{SI,D}}$ & Noiseless & at decoder & at rate $R_{0}$ & $I(Y^{n};S) < \epsilon$ \\
        \hline
        $\mathcal{R}_{\text{N,UC,J}}$ & Noisy BC  & - & unlimited & $I(Y^{n};\widetilde{Z}^{m}) < \epsilon$ \\
        \hline
        $\mathcal{R}_{\text{N,UC,S}}$ & Noisy BC  & - & unlimited & $I(Y^{n};\widetilde{Z}^{m}) < \epsilon$ \\
        \hline
        \end{tabular}
    \end{center}
\end{table*}
}

%
%
\section{Secure RDP over Channels}\label{sec:ISIT_results}
We first provide the capacity region for the secure RDP problem, depicted in Fig.~\ref{fig:simplified_SRDP}, in Theorem~\ref{theo:SRDPNoiseless}.
\begin{theorem}\label{theo:SRDPNoiseless}
    The region $\mathcal{R}$ for a given i.i.d source distribution $Q_{X}$ is the union overall joint probability distributions $P_{XUY}$, where $P_{Y} \approx Q_{X}$, of the set of all $(R, R_{0}, D)$ tuples that satisfy
    \begin{align}
        &R \geq I(U;X), \\
        &R_{0} \geq I(U;Y), \\
        &D \geq \mathbb{E}[d(X,Y)] \label{eq:distortionfromTheorem1}
    \end{align}
    such that $X-U-Y$ forms a Markov chain, and the cardinality $|\mathcal{U}|$ of $U$ can be limited to $|\mathcal{U}| \leq |\mathcal{X}|^{2} + 1$.
\end{theorem}

Comparing Theorem~\ref{theo:SRDPNoiseless} result with \cite[Theorem~6]{DenizRDP}, we observe that the amount of common randomness needed to achieve the realism and distortion constraints also suffices to achieve the secrecy constraint for this case.

%
%
\begin{IEEEproof}[Proof sketch for Theorem~\ref{theo:SRDPNoiseless}]
We next provide the achievability and converse proofs.

\emph{Proof sketch of achievability:} 
The method of output statistics of random binning (OSRB) \cite{OSRBAmin} is applied in the achievability proof of Theorem~\ref{theo:SRDPNoiseless}. The key idea of OSRB proofs is to show that the induced probability distributions of two coding schemes are asymptotically equal. This means that the rates that are later derived for one coding scheme, random binning that is called Protocol A in \cite{OSRBAmin}, will apply to the other coding scheme, random coding that is called Protocol B in \cite{OSRBAmin}. Specific choices of the encoder and decoder in these two coding, or binning, schemes are made in the ingress to the proofs of the theorems and corollaries in this paper to show that the induced probability distributions converge asymptotically, which follow below.

Fix a distribution $P_{XUY}$ such that $\mathbb{E}[d(X, Y)] \leq D + \epsilon_n$, in which $\epsilon_n > 0$ and $\epsilon_n\rightarrow 0$ when $n\rightarrow\infty$. Introduce a random variable $F$, drawn uniformly at random according to $P_{F} = \text{Unif}[1:2^{n\widetilde{R}}]$, representing the public choice of the encoder and decoder pair. Both the encoder and decoder have access to $F$. We generate the auxiliary random variable $U^{n}$ and assign bin indices $C$, $F$, and $S$ independently and uniformly to each sequence $u^{n} \in \mathcal{U}^{n}$. This generation and bin assignment is done in two ways: \\
(i) \emph{Random Binning} where the auxiliary random variable sequence $U^{n}$ is generated in an i.i.d manner according to $P_{U}^{n}$ such that the joint probability distribution factorizes as $P_{XU} = Q_{X}P_{U|X}$. The random bin assignment 

(1) assigns $C = \varphi_{1}(U^{n})$, where $\varphi_{1}:\mathcal{U}^{n} \rightarrow [1:2^{nR_{0}}]$ maps each sequence $u^{n}$ independently and uniformly at random to an index $P_{C|U^{n}} \sim \text{Unif}[1:2^{nR_{0}}]$; 

(2) assigns $F = \varphi_{2}(U^{n})$, where $\varphi_{2}:\mathcal{U}^{n} \rightarrow [1:2^{n\widetilde{R}}]$ maps each sequence $u^{n}$ independently and uniformly at random to an index $P_{F|U^{n}} \sim \text{Unif}[1:2^{n\widetilde{R}}]$; and 

(3) assigns $S = \varphi_{3}(U^{n})$, where $\varphi_{3}:\mathcal{U}^{n} \rightarrow [1:2^{nR}]$ maps each sequence $u^{n}$ independently and uniformly at random to an index $P_{S|U^{n}} \sim \text{Unif}[1:2^{nR}]$. 

The receiver recovers $U^{n}$ from $(C,F,S)$ using a decoder $P_{\hat{U}^{n}|CFS}^{\text{RB}}$ that produces the output $\hat{U}^{n}$ which is mapped to $Y^{n}$ according to $P^{\text{RB}}_{Y^{n}|\hat{U}^{n}}$; and\\
(ii) \emph{Random Coding}, where the common randomness $C$ and the selection of encoder-decoder pair $F$ are generated independently and uniformly at random according to $P_{C} \sim \text{Unif}[1:2^{nR_{0}}]$ and $P_{F} \sim \text{Unif}[1:2^{n\widetilde{R}}]$, respectively. The auxiliary random variable $U^{n}$ is generated according to the distribution $P_{U^{n}|CFX^{n}}^{\text{RB}}$, which is derived from the random binning and relates to the uniform random binning distributions by rewriting
\begin{equation}
    P_{CFU^{n}|X^{n}}^{\text{RB}} = P_{U^{n}|X^{n}} P_{C|U^{n}} P_{F|U^{n}} = P_{CF} P_{U^{n}|CFX^{n}}^{\text{RB}}.
\end{equation}
The selections of $S$ is the same as for the random binning, according to $P^{\text{RB}}_{S|U^{n}}$. Recovering $U^{n}$ is done with the same decoder $P_{\hat{U}^{n}|CFS}^{\text{RB}}$ as in the random binning with the output $\hat{U}^{n}$, which is mapped to $Y^{n}$ in the same way as in the random binning according to $P^{\text{RB}}_{Y^{n}|\hat{U}^{n}}$.

The two ways of assigning bin indices induce two different joint distributions, denoted $P^{\text{RB}}$ and $P^{\text{RC}}$ for the random binning and the random coding respectively.
\begin{align}
    P^{RB} :=& Q^{n}_{X} P_{U^{n}|X^{n}} P_{C|U^{n}} P_{F|U^{n}} P^{\text{RB}}_{S|U^{n}} P_{\hat{U}^{n}|CFS}^{\text{RB}} P^{\text{RB}}_{Y^{n}|\hat{U}^{n}} \nonumber \\  ~~~~~~=& P_{CFX^{n}} P_{U^{n}|CFX^{n}}^{\text{RB}} P_{S|U^{n}}^{\text{RB}} P^{\text{RB}}_{\hat{U}^{n}|CFS} P^{\text{RB}}_{Y^{n}|\hat{U}^{n}} \label{eq:prb_noiseless} \\
    P^{RC} :=& P_{C} P_{F} Q^{n}_{X} P_{U^{n}|CFX^{n}}^{\text{RB}} P^{\text{RB}}_{S|U^{n}} P^{\text{RB}}_{\hat{U}^{n}|CFS} P^{\text{RB}}_{Y^{n}|\hat{U}^{n}} \label{eq:prc_noiseless} .
\end{align}

The TV distance between $P^{\text{RB}}$ and $P^{\text{RC}}$ can be upper-bounded via the following triangle inequality
\begin{align}
    &||P^{RB} -P^{RC}||_{\text{TV}} \leq \nonumber \\
    &||P^{RB} \!-\! P_{C}P_{F}Q^{n}_{X}P_{U^{n}|CFX^{n}}^{\text{RB}} 
    P^{\text{RB}}_{S|U^{n}}P^{\text{RB}}_{\hat{U}^{n}|CFS}   \nonumber \\ & \qquad\qquad \cdot  P^{\text{RB}}_{Y^{n}|\hat{U}^{n}}||_{\text{TV}} \nonumber \\ 
    & \qquad\!+\! ||P_{C}P_{F}Q^{n}_{X}P_{U^{n}|CFX^{n}}^{\text{RB}}P^{\text{RB}}_{S|U^{n}}P^{\text{RB}}_{\hat{U}^{n}|CFS} P^{\text{RB}}_{Y^{n}|\hat{U}^{n}} \nonumber \\ &\qquad\qquad\qquad \!-\! P^{\text{RC}}||_{\text{TV}} \nonumber \\
    &\stackrel{(a)}{=} ||P^{RB} \!-\! P_{C}P_{F}Q^{n}_{X}P_{U^{n}|CFX^{n}}^{\text{RB}}P^{\text{RB}}_{S|U^{n}}P^{\text{RB}}_{\hat{U}^{n}|CFS}  \nonumber \\ & \qquad\qquad\qquad \cdot P^{\text{RB}}_{Y^{n}|\hat{U}^{n}}||_{\text{TV}} \nonumber \\ &\stackrel{(b)}{=} ||P_{CFX^{n}} \!-\!P_{C}P_{F}Q^{n}_{X}||_{\text{TV}} \label{eq:TVDsums}
\end{align}
where (a) follows from the definition of $P^{\text{RC}}$, (b) follows by \cite[Lemma 17]{CuffsThesis}. The TV distance in (\ref{eq:TVDsums}) asymptotically tend to zero if we impose the rate constraints given below.

First, the indices $C$ and $F$ are almost independent of $X^{n}$ and jointly uniformly distributed if we have \cite[Theorem~1]{OSRBAmin}
\begin{equation}
    \widetilde{R} + R_{0} < H(U|X). \label{eq:indep_of_S_C}
\end{equation}
The public message index $S$ and the public index $F$ are almost independent of $Y^{n}$ and jointly uniformly distributed if we have \cite[Theorem 1]{OSRBAmin}
\begin{equation}
    \widetilde{R} + R < H(U|Y), \label{eq:secrect_rdp}
\end{equation}
which ensures that the secrecy constraint in (\ref{eq:secrecy_SRDP}) is satisfied, as the alphabet $\mathcal{Y}=\mathcal{X}$ is finite, see also the statements from \cite[Eq.~2]{OSRBAmin} in \cite[Theorem~1]{OSRBAmin}.

Finally, $P_{U^{n}|CFX^{n}}^{\text{RB}}$, that is used in both $P^{\text{RC}}$ and $P^{\text{RB}}$, is fixed to be a Slepian-Wolf decoder $P^{\text{SW}}_{\hat{U}^{n}|CFS}$, which can recover $U^{n}$ from $(S,C,F)$ reliably with an error tends to zero as $n \rightarrow \infty$, if we have \cite[Lemma 1]{OSRBAmin}
\begin{equation}
    \widetilde{R} + R_{0} + R > H(U) \label{eq:reconstruction_rates},
\end{equation}
which makes the induced probability distribution of the decoder nearly equal
\begin{equation}
    P^{\text{RB}}_{\hat{U}^{n}|CFS} P_{Y^{n}|\hat{U}^{n}}^{\text{RB}} \approx P^{\text{RB}}_{U^{n}|CFS} P_{Y^{n}|U^{n}}^{\text{RB}} \mathbb{1}\{ \hat{U}^{n} = U^{n}\}
\end{equation}
Applying Fourier-Motzkin elimination \cite[Appendix D]{Elgamalbook} to (\ref{eq:indep_of_S_C})-(\ref{eq:reconstruction_rates}) results in the following achievable rates
\begin{align}
    &R = I(U;X) + \delta \label{eq:message_case1},\\
    &R_{0} = I(U;Y) + \delta \label{eq:common_rand_case1}
\end{align}
for any $\delta > 0$ that tends to zero as $n \rightarrow \infty$. 

Next, the TV distance between the desired output probability distribution in the realism constraint (\ref{eq:real_SRDP}) and the induced output probability distribution of the codebook $U^{n}(c,s)$ is
\begin{equation}
    \left|\left| \sum_{c,s} \frac{P_{Y^{n}|U^{n}}(y^{n}|U^{n}(c,s))}{2^{nR_{0}} 2^{nR}} - Q_{X}^{n} \right|\right|_{\text{TV}}.
\end{equation}
The distance between the assumed input probability distribution $Q_{X}^{n}$ and the induced input probability distribution of the codebook $U^{n}(c,s)$ in TV is
\begin{equation}
    \frac{1}{2^{nR_{0}}} \left|\left| \sum_{c,s} \frac{P_{X^{n}|U^{n}}(x^{n}|U^{n}(c,s))}{2^{nR}} - Q_{X}^{n} \right|\right|_{\text{TV}} \label{eq:induced_x_noiseless}
\end{equation}
where the probability distribution of the common randomness $C$ is moved outside the TV distance since it is independent of $X^{n}$. By the soft covering lemma \cite[Lemma~IV.1]{CuffChannelSynthesis}, the TV distances above go to zero if $R_{0} + R > I(U;Y)$ and $R > I(U;X)$. The two rates are satisfied by (\ref{eq:message_case1}) and (\ref{eq:common_rand_case1}), similar to the proof of Theorem 2 in \cite{Wagner2022rate}. 
Moreover, by the typical average lemma \cite[pp.~26]{Elgamalbook}, the distortion constraint (\ref{eq:dist_SRDP}) is satisfied since all sequence tuples $(x^n, u^n, y^n)$ are in the jointly typical set with high probability.

Equation (\ref{eq:TVDsums}) shows that the two induced distributions from the random binning and the random coding schemes are asymptotically equal, by \cite[2) in Lemma~3]{OSRBAmin} there exists an instance of $F=f$ such that when conditioning on this instance of $F=f$ the distance between $P^{\text{RB}}$ and $P^{\text{RC}}$ in (\ref{eq:TVDsums}) in TV is upper bounded as follows
\begin{align}
    &||P^{\text{RB}}_{CFX^{n}} - P^{\text{RC}}_{CFX^{n}}||_{\text{TV}} \leq \delta_{\text{TV}} \nonumber \\ 
    & \qquad\implies ||P^{\text{RB}}_{CX^{n}|F=f} - P^{\text{RC}}_{CX^{n}|F=f}||_{\text{TV}} \leq 2\delta_{\text{TV}} \label{eq:lemma3bound}
\end{align}
where $\delta_{\text{TV}} > 0$ vanishes when $n$ goes to infinity when we impose the rates in (\ref{eq:message_case1}) and (\ref{eq:common_rand_case1}). We can now identify $(P^{\text{RB}}_{U^{n}|CFX^{n}}(u^{n}|c,F=f,x^{n}), P^{\text{RB}}_{S|U^{n}}(s|u^{n}))$ as the encoder and $(P^{\text{SW}}_{\hat{U}^{n}|CFS}(\hat{u}^{n}|c,F=f,s), P_{Y^{n}|\hat{u}^{n}}(y^{n}|u^{n}))$ as the decoder, where the superscript SW means that it's the Slepian-Wolf decoder that was fixed earlier in the proof, to have an encoder-decoder design that achieves the constraints (\ref{eq:real_SRDP})-(\ref{eq:secrecy_SRDP}). 

\emph{Proof sketch for converse:} Assume there exists a tuple $(R, R_{0}, D)$ that satisfies (\ref{eq:real_SRDP})-(\ref{eq:secrecy_SRDP}) for some $\epsilon > 0$ and $n \geq 1$. Let $U_{i} \triangleq (S,C,X^{i-1})$ define an auxiliary random variable such that $X_i-U_i-Y_i$ forms a Markov chain. Let $T$ be a time-sharing random variable $T \in [1:n]$ that is independent of all other random variables and uniformly distributed over $[1:n]$.

We have
\begin{align}
    nR \geq& H(S|C) \geq I(X^{n};S|C) \stackrel{(a)}{=} I(X^{n};SC) \nonumber\\ 
    =& \sum^{n}_{i=1} I(X_{i};SC|X^{i-1})\stackrel{(b)}{=} \sum^{n}_{i=1} I(X_{i};SCX^{i-1})\nonumber\\
    \stackrel{(c)}{=}& nI(X_{T};U_{T}|T)= nI(X_{T};U_{T}T) \nonumber \\
    \stackrel{(d)}{=}& nI(X;U)
\end{align}
where (a) and (b) follow because $X^{n}$ is i.i.d and independent of $C$, (c) follows from the definition of $U_i$ and $T$, and (d) follows by defining $U \triangleq (U_{T}, T)$ and $X \triangleq X_{T}$.

We obtain
\begin{align}
    nR_{0} \geq& H(C|S) \geq I(Y^{n};C|S)\stackrel{(a)}{\geq} I(Y^{n};CS) - f(\epsilon) \nonumber \\ 
    =& \sum^{n}_{i=1} I(Y_{i};CS|Y^{i-1}) - f(\epsilon) \nonumber \\
    =& \sum^{n}_{i=1} [I(Y_{i};CSY^{i-1}) - I(Y_{i};Y^{i-1})] - f(\epsilon) \nonumber \\
    \stackrel{(b)}{\geq}& \sum^{n}_{i=1} [I(Y_{i};CSY^{i-1}) - g(\epsilon)] - f(\epsilon) \stackrel{(c)}\nonumber\\
    {\geq}& \sum^{n}_{i=1} [I(Y_{i};CSX^{i-1}) - g(\epsilon)] - f(\epsilon) \nonumber \\
    \stackrel{(d)}{=}& nI(Y_{T};U_{T}|T) - 2ng(\epsilon) - f(\epsilon)\nonumber\\ 
    {=}& nI(Y_{T};U_{T}T) - 2ng(\epsilon) - f(\epsilon) \nonumber \\
    \stackrel{(e)}{=}& nI(Y;U) - 2ng(\epsilon) - f(\epsilon)
\end{align}
where (a) follows for a function $f(\epsilon)$ that tends to zero as $\epsilon \rightarrow 0$, given by the secrecy constraint in (\ref{eq:secrecy_SRDP}), (b) follows by (\ref{eq:real_SRDP}) as $Y^{n}$ is nearly i.i.d and for a function $g(\epsilon)$ that tends to zero as $n \rightarrow \infty$, given by the realism constraint (\ref{eq:real_SRDP}), (c) follows by the data processing inequality applied to the Markov chain $(C,S,X^{i-1})-(Y^{i-1},C,S)-Y_{i}$, in (d) follows by the definition of $U_{i}$ and $T$, and because $Y^{n}$ is nearly i.i.d, and (e) follows by defining $U \triangleq (U_{T}, T)$ and $Y \triangleq Y_{T}$.

We have
\begin{align}
    D + \epsilon \stackrel{(a)}{\geq}& \mathbb{E}[d(X^{n},Y^{n})] = \frac{1}{n} \sum^{n}_{i=1} \mathbb{E}[d(X_{i},Y_{i})] \nonumber \\
    \stackrel{(b)}{\geq}& \mathbb{E}[d(X_{T},Y_{T})|T] + \epsilon' \nonumber \\
    \stackrel{(c)}{=}& \mathbb{E}[d(X,Y)] + \epsilon'\label{eq:dist_conv}
\end{align}
where (a) follows by (\ref{eq:dist_SRDP}), (b) follows since $Y^{n}$ is almost i.i.d for an $\epsilon'$ that goes to zero when $n \rightarrow \infty$ and (c) follows by defining $X \triangleq X_{T}$ and $Y \triangleq Y_{T}$.

Consider $\mathcal{X} \!\times\! \mathcal{Y} \!=\!\mathcal{X}\! \times\! \mathcal{X} \!=\! \{ (x_{1}, y_{1}), (x_{1},y_{2}), (x_{2}, y_{1}), (x_{2}, y_{2}), \dots , (x_{|\mathcal{X}|}, y_{|\mathcal{X}|}) \}$, where all tuples are indexed by $k \in [1:|\mathcal{X}|^{2}]$. Consider also the following $|X|^{2} + 1$ real-valued continuous functions on the connected and compact subset of all probability distributions on $\mathcal{X} \times \mathcal{X}$
\begin{align}
    f_{k}(P_{XY}) =& \begin{cases}
        P_{XY}((x_{i},y_{j})_{k}) & \text{for}~k \in [1 : |\mathcal{X}|^{2} - 1] \\
        H(X) & \text{for}~k = |\mathcal{X}|^{2} \\
        H(Y) & \text{for}~k = |\mathcal{X}|^{2} + 1 \\
    \end{cases}
\end{align}
where $i,j\in[1:|\mathcal{X}|]$. By the support lemma \cite[Lemma 15.4]{Csiszarbook}, there exists a random variable $U'$ that takes at most $|\mathcal{X}|^{2} + 1$ values such that we preserve $P_{XY}$, $H(X)$, $H(Y)$, $H(X|U)$, and $H(Y|U)$ if we replace $U$ with $U'$ so that the cardinality bound in Theorem~\ref{theo:SRDPNoiseless} follows.
\end{IEEEproof}

Next, we provide an inner bound for the secure noisy RDP region, depicted in Fig.~\ref{fig:noisy_SRDP_broadcast}.
\begin{theorem}\label{theo:SRDPnoisy}
    The rate region $\mathcal{R}_{N}$ for a given i.i.d source distribution $Q_{X}$ and a memoryless noisy BC $P_{\widetilde{Y} \widetilde{Z}|\widetilde{X}}$ includes the union overall joint probability distributions $P_{XW_1YW_2}$, where $P_{Y} \approx Q_{X}$, of the set of all $(R, R_{0}, D)$ tuples that satisfy (\ref{eq:distortionfromTheorem1}) and
    \begin{align}
        &I(W_{1};X) \leq R \leq I(W_{2};\widetilde{Y}), \\
        &R_{0} \geq I(W_{1};Y) + I(W_{2};\widetilde{Z}) - I(W_{2};\widetilde{Y})
    \end{align}
    such that $X\!-\!W_{1}\!-\!Y$ and $W_2\!-\!\widetilde{X}\!-\!(\widetilde{Y},\widetilde{Z})$ form Markov chains, and $(X, W_{1}, Y)$ is mutually independent of $(W_{2},\widetilde{X},\widetilde{Y},\widetilde{Z})$. It suffices to have $|\mathcal{W}_1|\!\leq\! |\mathcal{X}|^2\!+\!1$ and $|\mathcal{W}_2|\!\leq\! |\mathcal{\widetilde{X}}|\!+\!1$.
\end{theorem}

Compared to the related problem of secure strong coordination in \cite[Theorem~1]{cervia2020secure}, we observe that the first term $I(W_{1};Y)$ in the bound on common randomness is reduced, similarly to the reduction of common randomness between the RDP result in \cite[Theorem~6]{DenizRDP} and the distributed channel simulation result in \cite[Theorem~II.1]{CuffChannelSynthesis}. An important difference between our results and \cite{cervia2020secure} is that our secrecy constraint is imposed only the output of the decoder, rather than on both the input and the output. The proof of Theorem~\ref{theo:SRDPnoisy} follows similarly to the achievability proof of Theorem~\ref{theo:SRDPNoiseless} with the difference being that rates for the unsecured communication between the encoder and the decoder needs to be derived in order to show that the rates required to achieve the realism constraint is satisfied by the rates given by the initial part of the proof using OSRB. The proof of Theorem~\ref{theo:SRDPnoisy} is provided in Appendix~\ref{app:ThmSRDPNoisy}.

Moreover, under the assumption that $P_{\widetilde{Y}\widetilde{Z}|\widetilde{X}}$ a more capable BC as described in Definition~\ref{def:morecapabledef}, the region $\mathcal{R}_{N,MCB} \subseteq \mathcal{R}_{N}$ is next shown to be exact. Below, we provide the exact region for this special case.
\begin{corollary}\label{cor:SRDPmoreCapable}
    The region $\mathcal{R}_{N,MCB}$ for a given i.i.d source distribution $Q_{X}$ with the assumption that the decoder has a more-capable channel over the memoryless BC $P_{\widetilde{Y}\widetilde{Z}|\widetilde{X}}$ is the union overall joint probability distributions $P_{XWY}$, where $P_{Y} \approx Q_{X}$, of the set of all $(R,R_{0},D)$ tuples that satisfy (\ref{eq:distortionfromTheorem1}) and
    \begin{align}
        &I(W;X) \leq R\leq  I(\widetilde{X};\widetilde{Y}), \\
        &R_{0} \geq I(W;Y) + I(\widetilde{X};\widetilde{Z}) - I(\widetilde{X};\widetilde{Y})
    \end{align}
    where $X\!-\! W\! -\! Y$ forms a Markov chain, $(X,W,Y)$ and $(\widetilde{X},\widetilde{Y},\widetilde{Z})$ are mutually independent, and $|\mathcal{W}|\! \leq\! |\mathcal{X}|^{2} \!+\! 1$.
\end{corollary}

Corollary~\ref{cor:SRDPmoreCapable} result is similar to \cite[Corollary~1]{cervia2020secure}, but with a reduction in the required common randomness rate that corresponds to the secure communication rate of $I(\widetilde{X};\widetilde{Y}) - I(\widetilde{X};\widetilde{Z})$ between the encoder and the decoder. The proof for Corollary~\ref{cor:SRDPmoreCapable} is given below.

\begin{IEEEproof}[Proof sketch]
The achievability of Corollary~\ref{cor:SRDPmoreCapable} follows from the inner bound in Theorem~\ref{theo:SRDPnoisy} by setting $W_{1} = W$ and $W_{2} = \widetilde{X}$. The converse proof for Corollary~1 is given below.

Assume there exists a $(R, R_{0}, D)$ tuple that satisfied (\ref{eq:real_SRDP}), (\ref{eq:dist_SRDP}), and (\ref{eq:secrecy_noisyRDP}) for some $\epsilon > 0$ and $n \geq 1$. Define a time-sharing random variable $T \in [1:n]$ with distribution $T \sim \text{Unif}[1:n]$ that is independent of all other random variables in the model. Let $W_{t}=(C,\widetilde{Y}^{n})$ define the auxiliary random variable $W_{t}$.

We have 
\begin{align}
    0 \stackrel{(a)}{\leq}& I(\widetilde{X}^{n};\widetilde{Y}^{n}) - I(\widetilde{Y}^{n};X^{n}C) \nonumber \\
    \leq& I(\widetilde{X}^{n};\widetilde{Y}^n) - I(\widetilde{Y}^{n};X^{n}|C)  \nonumber\\ 
    \stackrel{(b)}{=}& nI(\widetilde{X};\widetilde{Y}) - \sum^{n}_{t=1} I(\widetilde{Y}^{n};X_{t}|X^{t-1}C) \nonumber \\
    \stackrel{(c)}{=}& nI(\widetilde{X};\widetilde{Y}) - \sum^{n}_{t=1} I(\widetilde{Y}^{n}X^{t-1}C;X_{t}) \nonumber \\ 
    \leq& nI(\widetilde{X};\widetilde{Y}) - \sum^{n}_{t=1} I(\widetilde{Y}^{n}C;X_{t}) \nonumber \\
    \stackrel{(d)}{=}& n[I(\widetilde{X};\widetilde{Y}) - I(W_{T};X_{T}|T)] \nonumber \\ 
    \stackrel{(e)}{=}& n[I(\widetilde{X};\widetilde{Y}) - I(W;X)] \label{eq:combination_ch_src}
\end{align}
where (a) follows from applying the data processing inequality in the Markov chain $(X^{n}, C)-\widetilde{X}^{n}-\widetilde{Y}^{n}$, (b) follows since the channel $P_{\widetilde{Y}\widetilde{Z}|\widetilde{X}}$ is memoryless, (c) follows since $X^{n}$ is i.i.d and independent of $C$, (d) follows from identifying $W_{T}$ and the definition of $T$, and (e) follows from defining $W \triangleq (W_{T}, T)$ and $X \triangleq X_{T}$.

We observe
\begin{align}
    nR_{0} \geq& H(C|\widetilde{Y}^{n})- H(C|\widetilde{Y}^{n}Y^{n}) = I(Y^{n};C|\widetilde{Y}^{n}) \nonumber \\
    =& I(Y^{n};C\widetilde{Y}^{n}) - I(Y^{n};\widetilde{Y}^{n}) \nonumber
\end{align}
\begin{align}
    \stackrel{(a)}{\geq}&\! I(Y^{n};C\widetilde{Y}^{n}) \!-\! I(Y^{n};\widetilde{Y}^{n}) \!+\! I(Y^{n};\widetilde{Z}^{n}) \!-\! f(\epsilon) \label{eq:main_cor_conv}
\end{align}
where (a) follows by the secrecy constraint (\ref{eq:secrecy_noisyRDP}) for a function $f(\epsilon)$ that tends to zero when $\epsilon$ tends to zero.

For the first term in (\ref{eq:main_cor_conv}), we have
\begin{align}
    I(Y^{n};C,\widetilde{Y}^{n}) =& \sum^{n}_{t=1} I(Y_{t};C\widetilde{Y}^{n}|Y^{t-1}) \nonumber \\
    \stackrel{(a)}{\geq}& \sum^{n}_{t=1} I(Y_{t};C\widetilde{Y}^{n}Y^{t-1}) - ng(\epsilon) \nonumber \\ 
    \geq& \sum^{n}_{t=1} I(Y_{t};C\widetilde{Y}^{n}) - ng(\epsilon) \nonumber \\
    \stackrel{(b)}{=}& nI(Y_{T};W_{T}|T) - 2ng(\epsilon) \nonumber \\ 
    =& n[I(Y_{T};W_{T}T) - I(Y_{T};T)] - 2ng(\epsilon) \nonumber \\ 
    \stackrel{(c)}{\geq}& nI(Y_{T};W_{T}T) - 3ng(\epsilon) \nonumber \\
    \stackrel{(d)}{=}& nI(Y;W) - 3ng(\epsilon)
\end{align}
where (a) follows by \cite[Lemma~3]{cervia2020secure} since $Y^{n}$ is nearly i.i.d by the realism constraint (\ref{eq:real_SRDP}), (b) follows by the definition of $W_{t}$ and $T$, (c) follows from \cite[Lemma~VI.3]{CuffChannelSynthesis} and by (\ref{eq:real_SRDP}), and (d) follows from defining $Y \triangleq Y_{T}$. 

For the second and third terms in (\ref{eq:main_cor_conv}), we have
\begin{align}
    &I(Y^{n};\widetilde{Y}^{n}) - I(Y^{n};\widetilde{Z}^{n}) \nonumber \\ 
    &\stackrel{(a)}{=} \sum^{n}_{t=1} [I(Y^{n};\widetilde{Y}_{t}|\widetilde{Y}^{t-1},\widetilde{Z}^{n}_{t+1}) - I(Y^{n};\widetilde{Z}_{t}|\widetilde{Y}^{t-1},\widetilde{Z}^{n}_{t+1})] \nonumber \\ 
    &\stackrel{(b)}{=} \sum^{n}_{t=1} [I(W_{2,t};\widetilde{Y}_{t}|W_{3,t}) - I(W_{2,t};\widetilde{Z}_{t}|W_{3,t})] \nonumber \\ 
    &\stackrel{(c)}{\leq} n[I(W_{2,T};\widetilde{Y}_{T}|W_{3,T},T) - I(W_{2,T};\widetilde{Z}_{T}|W_{3,T},T)] \nonumber \\ 
    &\qquad+ ng'(\epsilon) \nonumber \\ 
    &\stackrel{(d)}{=} n[I(W_{2};\widetilde{Y}_{T}|W_{3}) - I(W_{2};\widetilde{Z}_{T}|W_{3})] + ng'(\epsilon) \nonumber\\
    &\leq n \max_{t} [I(W_{2};\widetilde{Y}_{T}|W_{3}=t) - I(W_{2};\widetilde{Z}_{T}|W_{3}=t)] \nonumber \\ 
    &\qquad + ng'(\epsilon) \nonumber \\ 
    &\stackrel{(e)}{=} n[I(\widetilde{W}_{2};\widetilde{Y}) - I(\widetilde{W}_{2};\widetilde{Z})] + ng'(\epsilon) \nonumber \\ 
    &\stackrel{(f)}{\leq} n[I(\widetilde{X};\widetilde{Y}) - I(\widetilde{X};\widetilde{Z})] + ng'(\epsilon) \label{eq:last_conv_eq}
\end{align}
where (a) follows by Csiszár's sum identity \cite[pp.~25]{Elgamalbook} and the memorylessness of the channel $P_{\widetilde{Y}\widetilde{Z}|\widetilde{X}}$, (b) follows by letting $W_{3,t} \triangleq (\widetilde{Y}^{t-1}, \widetilde{Z}^{n}_{t+1})$ and $W_{2,t} \triangleq (Y^{n}, W_{3,t})$, (c) follows since $Y^{n}$ is almost i.i.d by (\ref{eq:real_SRDP}) and for a function $g'(\epsilon)$ such that $g'(\epsilon)$ tends to zero when $\epsilon$ tends to zero, (d) follows by letting $\widetilde{Y} \triangleq \widetilde{Y}_{T}$ and $\widetilde{Z} \triangleq \widetilde{Z}_{T}$, (e) follows from $\widetilde{W}_{2} \sim P_{W_{2}|W_{3}=t^{*}}$ (where $t^{*}$ is the maximizer) and that $\widetilde{W}_{2} - \widetilde{X} - (\widetilde{Y},\widetilde{Z})$ forms a Markov chain, and finally (f) follows from the assumption that the decoder has a more capable channel than the eavesdropper. Equations (\ref{eq:main_cor_conv})-(\ref{eq:last_conv_eq}) combines and results in the common randomness bound in Corollary~\ref{cor:SRDPmoreCapable}.

The bound on the distortion follows as in (\ref{eq:dist_conv}). Furthermore, similar to the converse for Theorem~\ref{theo:SRDPNoiseless}, the cardinality bound follows by the support lemma and can be limited to $|\mathcal{W}|\! \leq\! |\mathcal{X}|^{2} \!+\! 1$.
\end{IEEEproof}

%
%
\section{Secure RDP over Noiseless Channels with Side Information}\label{sec:WCNC_results}
In this section, we provide the rate regions for secure RDP with different availabilities of side information when the encoder and the decoder communicate over a noiseless channel. 

We first provide the secure RDP region with side information available at both the encoder and the decoder, depicted in Fig.~\ref{fig:SecRDPwithSIatBoth}.

\begin{theorem}\label{theo:bothSI}
    The region $\mathcal{R}_{\text{SI,ED}}$ is the union over all joint probability distributions $P_{XZUY}\! =\!Q_{XZ} P_{UY|XZ}$, where $P_{Y} \approx Q_{X}$, of the set of all $(R, R_{0}, D)$ tuples that satisfy (\ref{eq:distortionfromTheorem1}) and
    \begin{align}
        & R \geq I(U;X|Z), \label{eq:comRateSIRxTx} \\
        & R_{0} \geq I(U;Y) - I(U;Z), \label{eq:crRateSIRxTx} \\
        & R + R_{0} \geq I(U;Y|Z) - H(Z|Y) \label{eq:sumRateSIRxTx}
    \end{align}
    such that $X - (U, Z) - Y$ forms a Markov chain and the cardinality $|\mathcal{U}|$ of $U$ can be limited to $|\mathcal{X}|^{2}|\mathcal{Z}| + 2$.
\end{theorem}

The bounds on the communication rate (\ref{eq:comRateSIRxTx}) and the sum-rate (\ref{eq:sumRateSIRxTx}) in Theorem~\ref{theo:bothSI} are the same as in \cite[Theorem~8]{WagnerHamidGunduzarXiv2025} but our result differs in (\ref{eq:crRateSIRxTx}), which is not imposed in \cite[Theorem~8]{WagnerHamidGunduzarXiv2025}, due to the secrecy constraint. The achievability proof for Theorem~\ref{theo:bothSI} follows similarly to the achievability proof of Theorem~\ref{theo:SRDPNoiseless} and is provided in Appendix~\ref{app:thmBothSI}. The converse proof for Theorem~\ref{theo:bothSI} follows below.

\begin{IEEEproof}[Proof sketch for the converse proof for Theorem~\ref{theo:bothSI}]
Consider a tuple $(R, R_{0},D)$ that achieves (\ref{eq:real_SRDP}), (\ref{eq:dist_SRDP}), and (\ref{eq:secrecy_SRDPwithSI}). We define the auxiliary random variable $U_{i}\!=\!(S,C,Z_{\{n\} \backslash i})$ such that $X_{i}-(U_{i},Z_{i})-Y_{i}$ forms a Markov chain and a time-sharing random variable $T \sim \text{Unif}[1 : n]$, independent of other random variables. We have
\begin{align}
    nR_{0} \geq& H(C) - H(C|SY^{n}Z^{n}) + H(C|SY^{n}Z^{n}) \nonumber \\
    \stackrel{(a)}{=}& H(C) - H(Y^{n}|SCZ^{n}) - H(C|SZ^{n}) \nonumber \\
    &+ H(Y^{n}|SZ^{n}) + H(C|SY^{n}Z^{n}) \nonumber\\
    \stackrel{(b)}{=}& H(C) - H(Y^{n}|SCZ^{n}) - H(C|SZ^{n}) \nonumber \\
    &+ H(Y^{n}|SZ^{n}) + H(Z^{n}|SC) + H(C|S) \nonumber \\ 
    &+ H(Y^{n}|SCZ^{n}) - H(Z^{n}Y^{n}|S) \nonumber
\end{align}
\begin{align}
    \stackrel{(c)}{=}& H(C) - H(Y^{n}|SCZ^{n}) - H(C|SZ^{n}) \nonumber \\
    &+ H(Y^{n}|S) + H(Z^{n}|SY^{n}) - H(Z^{n}|S)   \nonumber \\
    &+ H(Z^{n}|SC)+ H(C|S) + H(Y^{n}|SCZ^{n}) \nonumber \\
    &- H(Z^{n}Y^{n}|S) \nonumber \\ 
    \stackrel{(d)}{=}& H(Y^{n}|S) - H(Y^{n}|SCZ^{n}) - H(Z^{n}|S)\nonumber \\
    &+ H(Z^{n}|SC)  + H(C) - H(C|SZ^{n}) + H(C|S) \nonumber \\ 
    &+ H(Z^{n}|SC) + H(Y^{n}|SCZ^{n}) - H(Y^{n}Z^{n}|S) \nonumber \\
    \stackrel{(e)}{=}& H(Y^{n}) - H(Y^{n}|SCZ^{n}) - H(Z^{n}|S) \nonumber \\
    &+ H(Z^{n}|SC) + H(C) - H(C|SZ^{n}) + H(C|S)  \nonumber \\ 
    &+ H(Y^{n}|SCZ^{n}) - H(Y^{n}Z^{n}|S) - g(\epsilon) \nonumber \\
    \stackrel{(f)}{=}& H(Y^{n}) - H(Y^{n}|SCZ^{n}) - H(Z^{n}|S) \nonumber \\
    &+ H(Z^{n}|SC) + H(C) - H(C|SZ^{n}) \nonumber\\
    &+ H(Y^{n}Z^{n}C|S) - H(Y^{n}Z^{n}|S) - g(\epsilon) \nonumber \\
    \geq& H(Y^{n}) - H(Y^{n}|SCZ^{n}) - H(Z^{n}) \nonumber \\
    &+ H(Z^{n}|SC) + H(C) - H(C|SZ^{n}) \nonumber\\
    &+ H(Y^{n}Z^{n}C|S) - H(Y^{n}Z^{n}|S) - g(\epsilon) \nonumber \\
    \geq& H(Y^{n}) - H(Y^{n}|SCZ^{n}) - H(Z^{n}) \nonumber \\
    &+ H(Z^{n}|SC) - g(\epsilon) \nonumber \\
    \stackrel{(g)}{\geq}& \sum_{i=1}^{n} [ H(Y_{i}) - H(Y_{i}|SCZ^{n}Y^{i-1})  - (H(Z_{i}) \nonumber \\
    &- H(Z_{i}|SCZ^{i-1})) - f(\epsilon)] - g(\epsilon) \nonumber \\
    \geq & \sum_{i=1}^{n} [H(Y_{i}) - H(Y_{i}|SCZ_{\{n\} \backslash i})  - (H(Z_{i}) \nonumber \\
    &- H(Z_{i}|SCZ_{\{n\} \backslash i})) - f(\epsilon)] - g(\epsilon) \nonumber \\
    =& n [H(Y_{T}|T) - H(Y_{T}|SCZ_{\{n\} \backslash T}T) \nonumber \\ 
    &- (H(Z_{T}|T) - H(Z_{T}|SCZ_{\{n\} \backslash T}T))] - nf(\epsilon) \nonumber \\
    &- g(\epsilon) \nonumber \\
    \stackrel{(h)}{=}& n[I(U;Y) - I(U;Z)] - nf(\epsilon) - g(\epsilon) \label{eq:convCR} 
\end{align}
where (a), (b), and (c) follow from the following chain rule expansions
\begin{align}
    &H(C|SY^{n}Z^{n}) = H(Y^{n}|SCZ^{n}) + H(C|SZ^{n}) \nonumber \\
    & \qquad \qquad \qquad \qquad- H(Y^{n}|SZ^{n}) \label{eq:chainRuleA}, \\
    &H(C|SZ^{n}Y^{n}) = H(Z^{n}|SC) + H(C|S) \nonumber \\
    & \qquad \qquad \qquad \qquad + H(Y^{n}|SCZ^{n}) \nonumber \\
    & \qquad \qquad \qquad \qquad - H(Z^{n}Y^{n}|S), \label{eq:chainRuleB} \\
    &H(Y^{n}|SZ^{n}) = H(Y^{n}|S) \!+\! H(Z^{n}|SY^{n}) \nonumber \\
    & \qquad \qquad \qquad \qquad \!-\! H(Z^{n}|S), \label{eq:chainRuleC}
\end{align}
with which we obtain all terms except $H(C)$, (d) follows from reordering the terms, (e) follows by the secrecy constraint (\ref{eq:secrecy_SRDPwithSI}) where $g(\epsilon)$ is a function that goes to zero as $n\! \rightarrow \! \infty$, (f) follows by applying the chain rule, (g) follows since $Z^{n}$ is i.i.d and because $Y^{n}$ is almost i.i.d by the realism constraint (\ref{eq:real_SRDP}) for a function $f(\epsilon)$ that tends to zero as $\epsilon \rightarrow 0$, and (h) follows by defining the random variables $Y\! =\! (Y_{T},T)$, $Z\! =\! (Z_{T},T)$, and identifying and defining the auxiliary random variable $U\!=\!(U_{T},T)$. 

The bound on $R$ and the bound on the sum rate constraint $R + R_{0}$ follows directly from \cite[Eq.~(23)-(27)]{WagnerHamidGunduzarXiv2025} with $[X]^{n}\! =\! X^{n}$, $Z^{n}\!=\!Z^{n}$, $M\!=\!S$, $J\!=\!C$, and $V_{T}\!=\!(S,C,Z_{\{n\} \backslash T})$. The distortion bound follows from (\ref{eq:dist_conv}). The cardinality bound follows by using the support lemma to preserve $P_{XZY}$ and the RHS in each bound in Theorem~\ref{theo:bothSI}.
\end{IEEEproof}

Next, we provide an inner bound for the rate region of secure RDP when side information is available at the decoder only, depicted in Fig.~\ref{fig:SecRDPwithSIatDec}. 
\begin{theorem}\label{theo:onlyDec}
    The region $\mathcal{R}_{\text{SI,D}}$ includes the union over all joint probability distributions $P_{XZUY}\!=\! Q_{XZ} P_{UY|XZ}$, where $P_{Y} \approx Q_{X}$, of the set of all $(R, R_{0}, D)$ tuples that satisfy (\ref{eq:distortionfromTheorem1}) and
    \begin{align}
        & R \geq I(U;X) - I(U;Z), \label{eq:comRateSITx} \\
        & R_{0} \geq I(U;Y) - I(U;Z), \label{eq:crRateSITx} \\
        & R + R_{0} \geq I(U;Y|Z) \label{eq:sumRateSITx} 
    \end{align}
    such that $X - (U,Z) - Y$ and $Z - X - U$ form Markov chains.
\end{theorem}
Comparing Theorem~\ref{theo:onlyDec} with \cite[Theorem~14]{WagnerHamidGunduzarXiv2025}, we observe that the communication rate (\ref{eq:comRateSITx}) is the same but the marginal constraint on the common randomness (\ref{eq:crRateSITx}) is the same as the sum-rate constraint in \cite[Theorem~14]{WagnerHamidGunduzarXiv2025}. This can again be interpreted as a consequence of the secrecy constraint. The proof of Theorem~\ref{theo:onlyDec} follows similarly to the achievability proof of Theorem~\ref{theo:SRDPNoiseless} and is provided in Appendix~\ref{app:thmOnlyDec}.

Moreover, under the assumption that the probability distribution of the side information and the output is jointly i.i.d, as in Definition~\ref{def:jointlyiid}, the rate region in Theorem~\ref{theo:onlyDec} can be shown to be exact, as established below.

\begin{corollary}\label{cor:onlyDecPairiid}
     The rate region $\mathcal{R}_{\text{SI,D}}$ when the decoder output and the side information are jointly i.i.d. is the union over all joint probability distributions $P_{XZUY}\!=\! Q_{XZ} P_{UY|XZ}$, where $P_{Y} \approx Q_{X}$, of the set of all $(R, R_{0}, D)$ tuples that satisfy (\ref{eq:distortionfromTheorem1}) and
    \begin{align}
        & R \geq I(U;X) - I(U;Z), \\
        & R_{0} \geq I(U;Y) - I(U;Z), \\
        & R + R_{0} \geq I(U;Y|Z)
    \end{align}
    such that $X - (U,Z) - Y$ and $Z - X - U$ form Markov chains and the cardinality $|\mathcal{U}|$ of $U$ can be limited to $|\mathcal{X}|^{2}|\mathcal{Z}| + 2$.
\end{corollary}

Corollary~\ref{cor:onlyDecPairiid} is different from the rate region established in \cite[Theorems~17 and 18]{WagnerHamidGunduzarXiv2025} where the jointly i.i.d nature of $Y^{n}$ and $Z^{n}$ is imposed as a constraint, while in this corollary we assume it to be a part of the model. The proof for the achievability of Corollary~\ref{cor:onlyDecPairiid} follows from Theorem~\ref{theo:onlyDec}. The converse proof follows below.

\begin{IEEEproof}[Proof sketch for Corollary~\ref{cor:onlyDecPairiid}]
Assume there exists a tuple $(R, R_{0}, D)$ such that (\ref{eq:real_SRDP}), (\ref{eq:dist_SRDP}), and (\ref{eq:secrecy_SRDPwithSI}) are satisfied for some $\epsilon\! >\! 0$ and $n\!\geq\! 1$ in the system shown in Fig.~\ref{fig:SecRDPwithSIatDec}. We further assume that the side information $Z^{n}$ is jointly i.i.d with the output of the decoder $Y^{n}$ as in the Definition~\ref{def:jointlyiid}. Let the auxiliary random variable $U_{i}\!=\!(S,C,Z_{\{n\} \backslash i})$ such that $X_{i}-(U_{i},Z_{i})-Y_{i}$ and $Z_{i} - X_{i} - U_{i}$ form Markov chains. Let $T\! \sim\! [1\!:\!n]$ be a time-sharing random variable independent of all other random variables. 
 
The converse bound on $R$ we use \cite[Eqs.~(23)-(25)]{WagnerHamidGunduzarXiv2025} with $[X]^{n} = X^{n}$, $Z^{n}\!=\!Z^{n}$, $M\!=\!S$, $J\!=\!C$, and $(V,T)\!=\!U_{T}\!=\!(S,C,Z_{\{n\} \backslash T})$. This yields a bound $I(U;X|Z)$ which is equal to $I(U;X) - I(U;Z)$ by the Markov chain $Z - X - U$. For the bound on the sum-rate $R\! +\! R_{0}$ follows from \cite[Eq.~(84)]{WagnerHamidGunduzarXiv2025} with the same relabeling of the random variables as for the converse for $R$. The common randomness rate follows from equation (\ref{eq:convCR}), the distortion follows as in (\ref{eq:dist_conv}), and the cardinality bound follows by using the support lemma.
\end{IEEEproof}

%
%
\section{Optimality of Separate Source-Channel Coding in Secure RDP}\label{sec:SCS_analysis}
Consider the exact rate region in Corollary~\ref{cor:SRDPmoreCapable} for the system model depicted in Fig.~\ref{fig:noisy_SRDP_broadcast}. The data processing inequality implies $R \geq I(W;X) \geq I(Y;X)$ for the Markov chain $X - W - Y$. Thus, with unlimited common randomness (i.e., $R_0\rightarrow\infty$), one can achieve $R = I(X;Y)$, which establishes the achievable rate for the region $\mathcal{R}_{\text{N,UC}}$ described in Definition~\ref{def:SecureRDPNoisy_SCS} as the set of all $(R,D)$ tuples that satisfy (\ref{eq:distortionfromTheorem1}) and $I(Y;X)\leq R\leq  I(\widetilde{X};\widetilde{Y})$ under the assumption that the decoder has a more-capable channel, $P_{Y} \approx Q_{X}$, and $(X,Y)$ and $(\widetilde{X},\widetilde{Y},\widetilde{Z})$ are mutually independent. 

We next prove that a separation-based source and channel code design can asymptotically achieve all rates such that $R \leq \kappa C_{\text{unsecure}}$, where $C_{\text{unsecure}}=\max_{P_{\widetilde{X}}} I(\widetilde{X};\widetilde{Y})$ is defined to be the \emph{unsecured} capacity of the channel between $\widetilde{X}$ and $\widetilde{Y}$ and $\kappa$ is a channel mismatch factor defined as $\kappa \triangleq \frac{m}{n}$. 

We next provide the relationship between the rate regions for separate source-channel coding and joint source-channel coding.

\begin{theorem}\label{theo:SCSepMoreCapable}
    The region $\mathcal{R}_{\text{N,UC,S}}$ for secure noisy RDP with separate source and channel coding, 
    is equivalent to the rate region $\mathcal{R}_{\text{N,UC,J}}$.
\end{theorem}

We remark that in \cite[Theorem~1~and~2]{tian2025source}, separate source and channel coding is shown to be optimal for the non-secure RDP problem when there is unlimited common randomness and for a more general block level perception metric. Similarly, for the strong coordination problem, it is shown in \cite[Section~V-D]{cervia2020strong} that with limited common randomness, separation of source and channel coding is not generally optimal.

\begin{IEEEproof}[Proof of $\mathcal{R}_{\text{N,UC,S}} \subseteq \mathcal{R}_{\text{N,UC,J}}$]
Consider a code with $(\kappa, R, D)$ in $\mathcal{R}_{\text{N,UC,S}}$. The probability distribution of the system depicted in Fig.~\ref{fig:noisy_SRDP_SCS} with a separate source-channel coding method is
\begin{align}
    &P_{X^{n}CS\widetilde{X}^{m}\widetilde{Y}^{m}\widetilde{Z}^{m}\hat{S}Y^{n}}
    = Q_{X}^{n} P_{C} P_{S|X^{n}C} P_{\widetilde{X}^{m}|SC} \nonumber \\
    & \qquad \qquad \qquad \qquad \qquad \cdot P_{\widetilde{Y}\widetilde{Z}|\widetilde{X}}^{m} P_{\hat{S}|\widetilde{Y}^{m}C} P_{Y^{n}|\hat{S}C}
\end{align}
where the encoder and decoder blocks have been separated into a source coding step $P_{S|X^{n}C}$ into a message $S$ and a channel coding step $P_{\widetilde{X}^{m}|SC}$, with the corresponding recovery of $\hat{S}$ by the channel decoder and reconstruction of $Y^{n}$ from $\hat{S}$ and $C$ at the source decoder. We see that these intermediate steps in the encoding and the decoding can be absorbed into one encoder $P_{\widetilde{X}^{n}|X^{n}C}$ and one decoder $P_{Y^{n}|\widetilde{Y}^{m}C}$, respectively. This yields the same structure on the probability distribution as that of a joint source-channel code, and thus $\mathcal{R}_{\text{N,UC,S}} \subseteq \mathcal{R}_{\text{N,UC,J}}$.

\emph{Proof of $\mathcal{R}_{\text{N,UC,J}} \subseteq \mathcal{R}_{\text{N,UC,S}}$}
Take a code with the achievable tuple $(\kappa, R, D)$ in $\mathcal{R}_{\text{N,UC,J}}$. The probability distribution of this system, depicted in Fig.~\ref{fig:noisy_SRDP_SCS} is
\begin{align}
    &P_{X^{n}C\widetilde{X}^{m}\widetilde{Y}^{m}\widetilde{Z}^{m}Y^{n}} = Q_{X}^{n} P_{C} P_{\widetilde{X}^{m}|X^{n}C} P_{\widetilde{Y}\widetilde{Z}|\widetilde{X}}^{m} \nonumber \\
    & \qquad \qquad \qquad \qquad \quad \cdot P_{Y^{n}|\widetilde{Y}^{m}C}. \label{eq:JSCCdist}
\end{align}
We have
\begin{align}
    I(X^{n};Y^{n}) \leq& I(X^{n};Y^{n}C) \stackrel{(a)}{=} I(X^{n};Y^{n}|C) \nonumber \\
    \stackrel{(b)}{\leq}& I(\widetilde{X}^{m};\widetilde{Y}^{m}|C)
    \leq I(\widetilde{X}^{m}C;\widetilde{Y}^{m}) \nonumber \\
    \stackrel{(c)}{=}& I(\widetilde{X}^{m};\widetilde{Y}^{m}) \stackrel{(d)}{=} \sum^{m}_{i=1} I(\widetilde{X}_{i};\widetilde{Y}_{i}) \label{eq:sourceCodeChannelCodeBound}
\end{align}
where (a) follows since $C$ is independent of $X^{n}$, (b) follows from the data processing inequality applied in the Markov chain $X^{n} - \widetilde{X}^{m} - \widetilde{Y}^{m} - Y^{n}$ for each $C=c$, (c) follows from the Markov chain $C - \widetilde{X}^{m} - \widetilde{Y}^{m}$, and (d) follows since $P_{\widetilde{Y}\widetilde{Z}|\widetilde{X}}^{m}$ is memoryless. For each $i$, we have
\begin{equation}
    I(\widetilde{X}_{i};\widetilde{Y}_{i}) < C_{\text{unsecure}}, \label{eq:capacityBound}
\end{equation}
which follows from the definition of $C_{\text{unsecure}}$. By combining (\ref{eq:sourceCodeChannelCodeBound}) and (\ref{eq:capacityBound}), we get
\begin{align}
    I(X^{n};Y^{n}) \leq& I(\widetilde{X}^{m};\widetilde{Y}^{m}) \leq \sum^{m}_{i=1} I(\widetilde{X}_{i};\widetilde{Y}_{i}) \nonumber \\ 
    \leq& mC_{\text{unsecure}} \label{eq:mLetterCapacityBound}
\end{align}

Furthermore, let $T\sim\text{Unif}[1:n]$ and independent of all other random variables and define $X\triangleq (X_{T},T)$ and $Y\triangleq (Y_{T},T)$. We have
\begin{align}
    I(X^{n};Y^{n}) =& \sum^{n}_{i=1} I(X_{i};Y^{n}|X^{i-1}) \nonumber \\
    \stackrel{(a)}{=}& \sum^{n}_{i=1} I(X_{i};Y^{n}X^{i-1}) = \sum^{n}_{i=1} I(X_{i};Y^{n}) \nonumber \\
    \geq& \sum^{n}_{i=1} I(X_{i};Y_{i}) = nI(X_{T};Y_{T}|T) \nonumber \\
    \stackrel{(b)}{=}& nI(X;Y) \label{eq:singleLetterXY}
\end{align}
where (a) follows since $X^{n}$ is i.i.d by the problem definition, (b) follows by the definitions of $X$ and $Y$. Combining (\ref{eq:mLetterCapacityBound}) and (\ref{eq:singleLetterXY}), we obtain
\begin{equation}
    I(X;Y) \leq \frac{m}{n} C_{\text{unsecure}} \stackrel{(a)}{\leq} (\kappa + \epsilon) C_{\text{unsecure}}, \label{eq:finalRate}
\end{equation}
where (a) follows from the channel mismatch factor $\kappa$ in Definition~\ref{def:SecureRDPNoisy_SCS} for some $\epsilon > 0$ that tends to zero as $n \rightarrow \infty$. 

From the outline of the achievable rate for the region $\mathcal{R}_{\text{N,UC}}$ in the beginning of this section, we see that a joint code can achieve $I(X;Y) \leq R \leq I(\widetilde{X};\widetilde{Y})$, which can then also achieve (\ref{eq:finalRate}). We will now show that a source code in the rate region $\mathcal{R}$ specified by Theorem~\ref{theo:SRDPNoiseless} with block-length $n$ concatenated with a channel code with block-length $m$ can achieve the rate of equation (\ref{eq:finalRate}) and fulfill Definition~\ref{def:SecureRDPNoisy_SCS}.

As discussed above, one can achieve $R\geq I(Y;X)$ under the assumptions considered in this section. We use this code with a long enough block length $n$ such that we can encode the source sequence $X^{n}$ onto the message $S$ at a rate $R = I(X;Y) + \delta$, for some $\delta >0$. This code achieves
\begin{align}
    &||P_{Y^{n}} - Q_{X}^{n}||_{\text{TV}} < \delta, \\
    &\mathbb{E}[d(X^{n},Y^{n})] \leq \mathbb{E}[d(X;Y)] + \delta, \\
    &I(Y^{n};S) \leq \delta
\end{align}
for a $\delta > 0$ that tends to zero as $n \rightarrow \infty$. For a large enough $m$ and an $\epsilon >0$ such that $m/n \leq \kappa + \epsilon$ we encode the message $S$ onto $\widetilde{X}^{m}$ and communicate it over the memoryless BC channel $P_{\widetilde{Y}\widetilde{Z}|\widetilde{X}}$ and produce the reconstruction $\hat{S}$ which is decoded with the source code above to produce $\hat{Y}^{n}$. The induced probability distribution of this concatenated source and channel code is
\begin{align}
P_{X^{n}CS\widetilde{X}^{m}\widetilde{Y}^{m}\widetilde{Z}^{m}\hat{S}\hat{Y}^{n}}
=&
Q_{X}^{n} P_{C} P_{S|X^{n}C} P_{\widetilde{X}^{m}|S} P_{\widetilde{Y}\widetilde{Z}|\widetilde{X}}^{m} \nonumber \\
& \quad \cdot P_{\hat{S}|\widetilde{Y}^{m}} P_{\hat{Y}^{n}|\hat{S}C}.
\label{eq:separateDist_corrected}
\end{align}
By (\ref{eq:finalRate}) and our chosen rate $R = I(X;Y) + \delta$ there exists a channel code at rate
\begin{equation}
    \frac{R - \delta}{\kappa} \leq C_{\text{unsecure}},
\end{equation}
such that we can decode $\hat{S} = S$ with error probability $P_{e} \triangleq P(\hat{S} \neq S) \rightarrow 0$ when $m \rightarrow \infty$. $\hat{S} \neq S$ is also the event where $\hat{Y}^{n} \neq Y^{n}$, since if $S = \hat{S}$, the decoder from Theorem 1 behaves as if $S$ were communicated over an ideal channel.

The expected distortion of this concatenated code is
\begin{align}
    \mathbb{E}[d(X^{n},\hat{Y}^{n})] \leq& (1-P_{e})(\mathbb{E}[d(X,Y)] + \delta) \nonumber \\ &+ P_{e}d_{\text{max}}
\end{align}
which follows since the alphabet $\mathcal{X}$ is finite. Furthermore, we have
\begin{align}
    ||P_{\hat{Y}^{n}} - Q_{X}^{n}|| \leq& ||P_{\hat{Y}^{n}} - P_{Y^{n}}||_{\text{TV}} + ||P_{Y^{n}} - Q_{X}^{n}||_{\text{TV}} \nonumber \\
    \leq& P_{e} + \delta
\end{align}
which follows since $P_{\hat{Y}^{n}} $ and $ P_{Y^{n}}$ only differ for $\hat{S} \neq S$ with probability $P_{e}$. 
We first observe that, since the channel encoder depends only on $S$, we have the Markov chain $Y^{n} - S - \widetilde{Z}^{m}$, and therefore we have
\begin{equation}
    I(Y^{n};\widetilde{Z}^{m}) \leq I(Y^{n};S) \leq \delta.
\end{equation}
Moreover, $\hat{Y}^{n}\neq Y^{n}$ implies $\hat{S}\neq S$, so we have $P(\hat{Y}^{n}\neq Y^{n}) \leq P_{e}\to 0$. Since the alphabets are finite, this implies 
\begin{equation}
\left| I(\hat{Y}^{n};\widetilde{Z}^{m}) - I(Y^{n};\widetilde{Z}^{m}) \right| \to 0 
\end{equation}
as $n,m\to\infty$. Hence, we asymptotically have $I(\hat{Y}^{n};\widetilde{Z}^{m}) \to 0$, which proves the secrecy requirement in Definition~\ref{def:SecureRDPNoisy_SCS}.

Finally, when $n,m \rightarrow \infty$, we have a concatenated code that fulfills the constraints in Definition~\ref{def:SecureRDPNoisy_SCS} with a rate of
\begin{equation}
    I(X;Y) \leq R \leq I(\widetilde{X};\widetilde{Y}) \leq \kappa C_{\text{unsecure}}.
\end{equation}
This is the same as the joint source-channel coding derived from Corollary~\ref{cor:SRDPmoreCapable}. Hence, we have $\mathcal{R}_{\text{N,UC,J}} \subseteq \mathcal{R}_{\text{N,UC,S}}$.
\end{IEEEproof}

%
%


%
%
\section{Example}\label{sec:examples}
\subsection{Binary Example for Theorem~\ref{theo:SRDPNoiseless}}
The first example is an evaluation of the rate region in Theorem~\ref{theo:SRDPNoiseless} with Hamming distortion as the distortion measure, whose rate region is denoted as $\mathcal{R}_{\text{BSC},\text{B}}$. The system model is depicted in Fig.~\ref{fig:simplified_SRDP} and let $X \sim \text{Bern}(0.5)$. We provide an achievable rate region for this example below.

\begin{corollary}\label{cor:simpleEx}
    The region $\mathcal{R}_{\text{BSC},\text{B}}$ includes the union over all joint probability distributions $P_{XUY}=Q_XP_{UY|X}$ of the set of all $(R, R_{0}, D)$ tuples and $\alpha,\beta\! \in\! [0,1]$ that satisfy
    \begin{align}
        & R \geq 1 - H_{b}(\alpha), \\
        & R_{0} \geq 1 - H_{b}(\beta),\\
        & D \geq \alpha * \beta
    \end{align}
    such that $X - U - Y$ forms a Markov chain.
\end{corollary}
The proof for Corollary~\ref{cor:simpleEx} is provided in Appendix~\ref{app:corSimpleEx}.


\begin{figure}[t]
    \centering
    \includegraphics[width=0.47\textwidth]{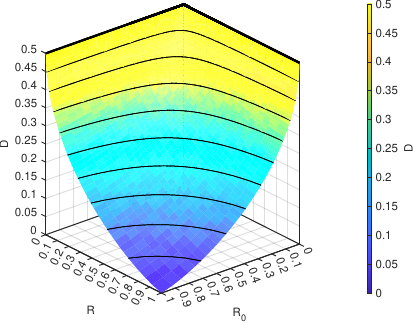}
    \caption{The achievable rate region's boundary for the communication rate $R$, common randomness rate $R_{0}$, and distortion $D$ with strong secrecy and realism defined in Definition~\ref{def:secRDPwithSI}. For distortions $>0.5$, the achievable region extends along the distortion (vertical) axis beyond $D = 0.5$ for all $R \geq 0 $ and $ R_{0} \geq 0$.}
    \label{fig:r_r0_d_trade_off_thm1_ISIT}
\end{figure}

Fig.~\ref{fig:r_r0_d_trade_off_thm1_ISIT} shows the region described in Corollary~\ref{cor:simpleEx} where in the direction towards the viewer, $R$ and $R_{0}$ increases, and in the vertical direction the distortion $D$ increases. The figure indicates that, for a fixed communication rate, the distortion decreases as the rate of common randomness grows. Similarly, when the common randomness rate is held constant, higher communication rates also yield lower distortion. This reflects a clear trade‑off between the two rates, resembling the behavior observed in \cite[Fig.~2]{CuffChannelSynthesis}, despite the absence of an explicit sum‑rate requirement. For example, near a distortion level of 0.1 and in the regime of small common‑randomness rates, the communication rate can be reduced by approximately 45-52\% when the common randomness is increased by 40-87\%. At distortion levels near 0.4, and when the common randomness rate is of similar magnitude as the communication rate, reductions of about 31-39\% in communication rate can be achieved with 43–63\% increases in common randomness. These observations highlight the significant improvements in the communication rate enabled by common randomness in secure RDP settings, improvements that do not arise in standard RD formulations.

\subsection{Gaussian Source Example for the Decoder-Only Side-Information Setting}

We remark that Theorem~\ref{theo:onlyDec} and Corollary~\ref{cor:onlyDecPairiid} are stated for finite alphabets and a near-perfect realism constraint (\ref{eq:real_SRDP}). For the Gaussian example below, we therefore proceed as follows. First, we consider the corresponding continuous-alphabet decoder-only side-information model and impose the same Markov chains as in Corollary~\ref{cor:onlyDecPairiid}. Second, we assume perfect marginal realism, i.e., $P_{Y}=P_{X}$, which is justified for Gaussian sources under mean-squared error distortion by the equivalence between the RDP regions for near-perfect and perfect marginal realism, as established in \cite[Claim~6 and Theorem~7]{WagnerHamidGunduzarXiv2025}. Motivated by these, we next evaluate a jointly Gaussian single-letter family that satisfies the Markov chains in Corollary~\ref{cor:onlyDecPairiid}, which yields an achievable Gaussian specialization of the single-letter expressions in Corollary~\ref{cor:onlyDecPairiid}. 

We consider the normalized Gaussian model
\begin{equation}
    (X,Z) \sim \mathcal{N}\!\left(0,
    \begin{bmatrix}
        1 & \eta \\
        \eta & 1
    \end{bmatrix}\right), \qquad |\eta|<1,
    \label{eq:gaussXZmodel}
\end{equation}
and the squared-error distortion measure $d(x,y)=(x-y)^{2}$. For a target distortion level $\Delta\in(0,2]$, define $\rho_{\Delta} \triangleq 1-\frac{\Delta}{2}$. Under perfect marginal realism and $Y\sim X\sim \mathcal{N}(0,1)$, we have
\begin{align}
    \mathbb{E}[(X-Y)^{2}] =& \mathbb{E}[X^{2}] + \mathbb{E}[Y^{2}] - 2\mathbb{E}[XY] \nonumber \\
    =& 2 - 2\mathbb{E}[XY]. \label{eq:gaussDistIdentity}
\end{align}
Therefore, the equality $\mathbb{E}[(X-Y)^{2}] = \Delta$ is equivalent to imposing $\mathbb{E}[XY]=\rho_{\Delta}$.

For any $\nu \in (\rho_{\Delta}^{2},1)$, define
\begin{equation}
    s(\nu) \triangleq \frac{(1-\nu)(1-\eta^{2})}{\nu-\eta^{2}}. \label{eq:sOfNu}
\end{equation}

\begin{proposition}\label{prop:GaussianSingleLetterFamily}
Fix $\Delta\in(0,2-2|\eta|]$ and let
\begin{equation}
    U = X + N_{U}, \qquad N_{U}\sim \mathcal{N}(0,s(\nu))
    \label{eq:gaussAuxiliaryU}
\end{equation}
where $N_{U}$, which is Gaussian distributed, is independent of $(X,Z)$. Next, define
\begin{equation}
    S \triangleq \mathbb{E}[X|U,Z].
\end{equation}
Then, $S$ is Gaussian and satisfies
\begin{equation}
    S = \frac{1-\eta^{2}}{1+s(\nu)-\eta^{2}}\,U
    + \frac{\eta\,s(\nu)}{1+s(\nu)-\eta^{2}}\,Z
    \label{eq:SLinearForm}
\end{equation}
with variance $\mathrm{Var}(S)=\nu$. Finally, define
\begin{equation}
    Y = \frac{\rho_{\Delta}}{\nu}S + N_{Y}
    \label{eq:gaussYConstruction}
\end{equation}
where
\begin{equation}
    N_{Y}\sim \mathcal{N}\!\left(0,1-\frac{\rho_{\Delta}^{2}}{\nu}\right)
    \label{eq:gaussYNnoise}
\end{equation}
is independent of $(X,Z,U,S)$. Since $\Delta\leq 2-2|\eta|$ implies $\rho_{\Delta}\geq |\eta|$, any $\nu\in(\rho_{\Delta}^{2},1)$ satisfies $\nu>\eta^{2}$. Hence, $s(\nu)>0$ and $1-\rho_{\Delta}^{2}/\nu>0$, so \eqref{eq:gaussAuxiliaryU} and \eqref{eq:gaussYNnoise} are well defined.

Then, the following hold:
\begin{align}
    &Z-X-U, \label{eq:gaussMC1}\\
    &X-(U,Z)-Y, \label{eq:gaussMC2}\\
    &P_{Y}=P_{X}=\mathcal{N}(0,1), \label{eq:gaussRealism}\\
    &\mathbb{E}[(X-Y)^{2}] = \Delta. \label{eq:gaussExactDist}
\end{align}
Moreover, the three single-letter quantities in Corollary~\ref{cor:onlyDecPairiid} evaluate to
\begin{align}
    R_{\mathrm{G},1}(\nu)
    &\triangleq I(U;X)-I(U;Z)
    = I(U;X|Z) \nonumber \\
    &= \frac{1}{2}\log\frac{1-\eta^{2}}{1-\nu}, \label{eq:R1Gaussian} \\
    R_{\mathrm{G},2}(\nu)
    &\triangleq I(U;Y)-I(U;Z) \nonumber \\
    &= \frac{1}{2}\log\frac{1+s(\nu)-\eta^{2}}{1+s(\nu)-\rho_{\Delta}^{2}/\nu^{2}}, \label{eq:R2Gaussian} \\
    R_{\mathrm{G},3}(\nu)
    &\triangleq I(U;Y|Z)
    = \frac{1}{2}\log\frac{\nu^{2}-\eta^{2}\rho_{\Delta}^{2}}{\nu(\nu-\rho_{\Delta}^{2})}. \label{eq:R3Gaussian}
\end{align}
Consequently, the jointly Gaussian family above yields the explicit one-parameter region
\begin{align}
    &R \geq R_{\mathrm{G},1}(\nu), \nonumber \\
    &R_{0} \geq R_{\mathrm{G},2}(\nu), \nonumber \\
    &R+R_{0} \geq R_{\mathrm{G},3}(\nu). \label{eq:GaussFamilyRegion}
\end{align}
\end{proposition}
The proof of Proposition~\ref{prop:GaussianSingleLetterFamily} is provided in Appendix~\ref{app:propGaussian}.

\begin{remark}\label{rem:GaussianBenchmark}
For $0<\Delta<2-2|\eta|$, the communication term in \eqref{eq:R1Gaussian} satisfies
\begin{equation}
    \lim_{\nu\downarrow \rho_{\Delta}^{2}} R_{\mathrm{G},1}(\nu)
    = \frac{1}{2}\log\frac{1-\eta^{2}}{1-\rho_{\Delta}^{2}}.
    \label{eq:GaussianBenchmarkLimit}
\end{equation}
The same expression is the exact minimum communication rate in the corresponding \emph{non-secure} RDP Gaussian side-information problem with finite-rate common randomness and perfect or near-perfect marginal realism \cite[Proposition~20]{WagnerHamidGunduzarXiv2025}. Moreover, we have
\begin{equation}
    \lim_{\nu\downarrow \rho_{\Delta}^{2}} R_{\mathrm{G},3}(\nu)=\infty.
\end{equation}
Hence, within the jointly Gaussian family in Proposition~\ref{prop:GaussianSingleLetterFamily}, approaching the minimum-communication boundary requires an arbitrarily large common-randomness rate.

At the boundary $\Delta=2-2|\eta|$, we have $\rho_{\Delta}=|\eta|$, and thus we obtain
\begin{equation}
    \lim_{\nu\downarrow \rho_{\Delta}^{2}} R_{\mathrm{G},1}(\nu)=0.
\end{equation}
Also, we have
\begin{equation}
    \lim_{\nu\downarrow \rho_{\Delta}^{2}} R_{\mathrm{G},3}(\nu)=
    \begin{cases}
        0, & \eta=0,\\[0.8ex]
        \frac{1}{2}\log 2, & \eta\neq 0.
    \end{cases}
\end{equation}
\end{remark}

\begin{remark}\label{rem:GaussianZeroRate}
For $\Delta \geq 2-2|\eta|$, zero communication is feasible in the perfect-marginal-realism Gaussian model. If $\eta>0$, choosing $Y=Z$ gives $Y\sim \mathcal{N}(0,1)$ and
\begin{equation}
    \mathbb{E}[(X-Y)^{2}] = \mathbb{E}[(X-Z)^{2}] = 2-2\eta.
\end{equation}
If $\eta<0$, choosing $Y=-Z$ gives $Y\sim \mathcal{N}(0,1)$ and
\begin{equation}
    \mathbb{E}[(X-Y)^{2}] = \mathbb{E}[(X+Z)^{2}] = 2+2\eta = 2-2|\eta|.
\end{equation}
When $\eta=0$, choosing $Y\sim\mathcal{N}(0,1)$ independent of $Z$ yields distortion $2$. Therefore, zero communication is feasible whenever
\begin{equation}
    \Delta \geq \Delta_{0}(\eta)\triangleq 2-2|\eta|. \label{eq:ZeroRateThresholdGaussian}
\end{equation}
\end{remark}

Thus, the Gaussian family in Proposition~\ref{prop:GaussianSingleLetterFamily} shows explicitly how the decoder side information quality, captured by $|\eta|$, influences the secure single-letter trade-off. In particular, stronger side information (i.e., higher correlation with $X$) decreases the communication term in \eqref{eq:R1Gaussian} and reduces the zero-rate achievability threshold in \eqref{eq:ZeroRateThresholdGaussian}.

\section{Conclusion}\label{sec:conclusion}
We studied secure RDP tradeoffs under negligible information leakage for both noiseless communication channels and noisy broadcast channels. We characterized the exact secure RDP region for the noiseless case. For transmission over broadcast channels, we derived an inner bound and proved its tightness for a class of more-capable broadcast channels. For the latter exact secure RDP region, we also proved that separating source coding and channel coding is optimal when the encoder and the decoder have access to an unlimited amount of common randomness. We also established the exact RDP region when both the encoder and the decoder have access to side information correlated with the source and the channel is noiseless. When only the decoder have access to correlated side information and the channel is noiseless, we derived a general inner bound and identified a special case for which the secure RDP region is exact. Moreover, our binary and Gaussian examples showed that common randomness can substantially reduce the communication rate in secure RDP settings, a gain that standard rate-distortion settings cannot attain.

These results were argued to be relevant for applications such as neural image compression, where systems must maintain perceptual quality while limiting information leakage when the encoder output is transmitted over public channels. Our analysis established the fundamental information-theoretic limits for trustworthy learned compression methods. In future, we will design trustworthy neural compression methods to approach these information-theoretic limits.

\section*{Acknowledgments} 
The authors used ChatGPT and Copilot to mainly revise and improve parts of the text. All content was reviewed and edited by the authors, who assume full responsibility. This work was partially supported by the ZENITH Research \& Leadership Fund, the Swedish Foundation for Strategic Research (SSF), and the German Federal Ministry of Research, Technology and Space (BMFTR) 6GEM+ Transfer Hub under Grants 16KIS2412 and 16KISS005.

\begin{appendices}

%
%
\section{Proof of Theorem~\ref{theo:SRDPnoisy} for Secure RDP over a Noisy Broadcast Channel}\label{app:ThmSRDPNoisy}
\begin{IEEEproof}[Proof sketch]
Similar to Theorem~\ref{theo:SRDPNoiseless}, the achievability proof for Theorem~\ref{theo:SRDPnoisy} applies the Output Statistics of Random Binning (OSRB) method \cite{OSRBAmin}.

Let $P_{XWY}$ be any distribution such that $\mathbb{E}[d(X^{n},Y^{n})] \leq D + \epsilon_{n}$, where $\epsilon_{n} \rightarrow 0$ as $n \rightarrow \infty$. We also introduce the random variable $F \sim \text{Unif}[1:2^{n\widetilde{R}}]$ that is available at the encoder and decoder and represents the randomly chosen encoder-decoder pair.

Similarly to the proof of Theorem~\ref{theo:SRDPNoiseless}, we generate the auxiliary random variable and assign bin indices $C$ and $F$ to each sequence $w^{n}$ in two ways:

(i) \emph{random binning} where the auxiliary random variable sequence $W^{n}$ is generated i.i.d according to $P_{W}^{n}$ such that the joint distribution $P_{X^{n}W^{n}}$ factorizes as $P_{X^{n}W^{n}} = Q_{X}^{n} P_{W^{n}|X^{n}}$ and $P_{W^{n}\widetilde{X}^{n}\widetilde{Y}^{n}\widetilde{Z}^{n}}$ factorizes as $P_{W^{n}\widetilde{X}^{n}\widetilde{Y}^{n}\widetilde{Z}^{n}} = P_{W}^{n} P_{\widetilde{X}^{n}|W^{n}} P_{\widetilde{Y}\widetilde{Z}|\widetilde{X}}^{n}$ where $P_{\widetilde{Y}\widetilde{Z}|\widetilde{X}}^{n}$ is the memoryless BC given in the system model depicted in Fig.~\ref{fig:noisy_SRDP_broadcast}. We randomly assign each $w^{n}$ bin indices $C = \varphi_{1}(w^{n})$ and $F = \varphi_{2}(w^{n})$ such that the decoder can reliably recover $W^{n}$ from $C$, $F$ and the broadcast channel (BC) output $\widetilde{Y}$. The binning functions $\varphi_{1}: \mathcal{W}^{n} \rightarrow [1:2^{nR_{0}}]$ and $\varphi_{2} : \mathcal{W}^{n} \rightarrow [1:2^{n\widetilde{R}}]$ assigns a bin index to each $w^{n}$ uniformly at random. The encoder selects the channel input according to $P^{\text{RB}_{\widetilde{X}^{n}|W^{n}X^{n}}}$. The decoder $P_{\hat{W}^{n}|CF\widetilde{Y}^{n}}^{\text{RB}}$ recovers $W^{n}$ from $(C,F,\widetilde{Y})$ and the output is $\hat{W}^{n}$. $Y^{n}$ is produced from $\hat{W}^{n}$ with $P^{\text{RB}}_{Y^{n}|\hat{W}^{n}}$; and

(ii) \emph{random coding} where the indices $C$ and $F$ for the common randomness and the choice of encoder-decoder are randomly and uniformly generated according to $P_{C}$ and $P_{F}$. The auxiliary random variable $W^{n}$ is generated according to $P^{RB}_{W^{n}|CFX^{n}}$, which is the same distribution as in the random binning and the relation is derived as in the proof for Theorem~\ref{theo:SRDPNoiseless}. The auxiliary random variable $W^{n}$ is recovered from $(C,F,\widetilde{Y}^{n})$ using the same decoder $P_{\hat{W}^{n}|CF\widetilde{Y}^{n}}^{\text{RB}}$ as in the random binning with the output $\hat{W}^{n}$. Similarly, $Y^{n}$ is produced from $\hat{W}^{n}$ with $P^{\text{RB}}_{Y^{n}|\hat{W}^{n}}$, which is the same as in the random binning.

The random binning and random coding ways of assigning bin indices $C$ and $F$ to the generated auxiliary random variable $W^{n}$ induces two probability distributions denoted $P^{\text{RB}}$ and $P^{\text{RC}}$, respectively. The induced distributions are
\begin{align}
    P^{\text{RB}} =& Q_{X}^{n} P_{W^{n}|X^{n}} P_{C|W^{n}} P_{F|W^{n}} P_{\widetilde{X}^{n}|W^{n}}^{\text{RB}} P_{\widetilde{Y}^{n}\widetilde{Z}^{n}|\widetilde{X}^{n}} \nonumber \\ 
    &\cdot P^{\text{RB}}_{\hat{W}^{n}|CF\widetilde{Y}^{n}} P^{\text{RB}}_{Y^{n}|\hat{W}^{n}} \nonumber \\ 
    =& P_{CFX^{n}} P^{\text{RB}}_{W^{n}|CFX^{n}} P_{\widetilde{X}^{n}|W^{n}}^{\text{RB}} P_{\widetilde{Y}^{n}\widetilde{Z}^{n}|\widetilde{X}^{n}} \nonumber \\
    &\cdot P^{\text{RB}}_{\hat{W}^{n}|CF\widetilde{Y}^{n}} P^{\text{RB}}_{Y^{n}|\hat{W}^{n}} , \label{eq:rb_dist_case2} \\
    P^{\text{RC}} =& P_{C} P_{F} Q_{X}^{n} P^{\text{RB}}_{W^{n}|CFX^{n}} P_{\widetilde{X}^{n}|W^{n}}^{\text{RB}} P_{\widetilde{Y}^{n}\widetilde{Z}^n|\widetilde{X}^{n}} \nonumber \\
    &\cdot P^{\text{RB}}_{\hat{W}^{n}|CF\widetilde{Y}^{n}} P^{\text{RB}}_{Y^{n}|\hat{W}^{n}}. \label{eq:rc_dist_case2}
\end{align}

In the same way as in equation (\ref{eq:TVDsums}), the TV distance between $P^{\text{RB}}$ and $P^{\text{RC}}$ can be bounded using the triangle inequality which yields
\begin{align}
    &||P^{\text{RB}} - P^{\text{RC}}||_{\text{TV}} \leq  ||P_{CFX^{n}} - P_{C}P_{F}Q_{X}^{n} ||_{\text{TV}}. \label{eq:tvd_noisy}
\end{align}

From (\ref{eq:tvd_noisy}) we have that the TV distance between the induced distributions from random binning and random coding vanishes if the distributions of $(C,F)$ are uniform and independent of $X^{n}$ in the random binning. By \cite[Theorem 1]{OSRBAmin} $C$ and $F$ are nearly jointly uniformly distributed and nearly independent of $X^{n}$ if we have the following rate
\begin{equation}
    \widetilde{R} + R_{0} < H(W|X) \label{eq:indep_X_of_F}.
\end{equation}

$W^{n}$ can be reliably recovered from $(C,F,\widetilde{Y}^{n})$ if we fix the decoder $P^{\text{RB}}_{\hat{W}^{n}|CF\widetilde{Y}^{n}}$ to be a Slepian-Wolf decoder. By \cite[Lemma~1]{OSRBAmin} the Slepian-Wolf decoder succeeds if we have the rate
\begin{equation}
    \widetilde{R} + R_{0} > H(W|\widetilde{Y}). \label{eq:reliable_decoding}
\end{equation}

By \cite[Theorem~1]{OSRBAmin}, we can ensure that the introduced random variable $Y^{n}$ and $\widetilde{Z}^{n}$ are nearly independent of each other and of $F$ if we have the following rate
\begin{equation}
    \widetilde{R} < H(W|Y,\widetilde{Z}) \label{eq:f_indep_prob}
\end{equation}
which ensures that the secrecy constraint (\ref{eq:secrecy_noisyRDP}) is satisfied since the alphabets of $Y^{n}$ and $Z^{n}$ are finite.

Let the auxiliary random variable $W = (W_{1}, W_{2})$ such that $(X,W_{1},Y)$ and $(W_{2},\widetilde{X},\widetilde{Y},\widetilde{Z})$ are mutually independent. By equations (\ref{eq:reliable_decoding}) and (\ref{eq:f_indep_prob}), we have
\begin{align}
    &R_{0} > H(W|\widetilde{Y}) - H(W|Y,\widetilde{Z}) \nonumber \\ 
    &= I(W_{1},W_{2};Y,\widetilde{Z}) - I(W_{1},W_{2};\widetilde{Y}) \nonumber  \\ 
    &= I(W_{1};Y) + I(W_{2};\widetilde{Z}) - I(W_{2};\widetilde{Y}) \label{eq:proto_coding_ineq}.
\end{align}

Additionally, from equations (\ref{eq:indep_X_of_F}) and (\ref{eq:reliable_decoding}), we have
\begin{equation}
    I(W;\widetilde{Y}) = I(W_{2};\widetilde{Y}) > I(W_{1};X) = I(W;X). \label{eq:jscc_case2}
\end{equation}

To prove the achievability of the realism constraint in (\ref{eq:real_SRDP}), we introduce the rate $R_{J}$ and let that represent the rate of reliable communication between the encoder and the decoder with the auxiliary random variable $W_{2}$ forming the Markov chain $W_{2}^{n} - \widetilde{X}^{n} - (\widetilde{Y}^{n},\widetilde{Z}^{n})$. This communication rate $R_{J}$ does not consider any secrecy constraint. The distance in TV between the induced output probability distribution from the codebook of $W_{1}(c,j)$, by the common randomness $C$ at rate $R_{0}$ and the reliable communication at rate $R_{J}$, and the target output probability distribution in (\ref{eq:real_SRDP}) is
\begin{equation}
    \left| \left| \sum_{c,j} \frac{P_{Y^{n}|W_{1}^{n}}(y^{n}|W_{1}(c,j))}{2^{nR_{j}} 2^{nR_{0}}} - Q_{X}^{n} \right| \right|_{\text{TV}}. \label{eq:TVD_input_noisy}
\end{equation}

The distance in TV between the induced input probability distribution of the codebook $W_{1}(c,j)$ and the input probability distribution given in the problem is
\begin{equation}
    \frac{1}{2^{nR_{0}}}\left| \left| \sum_{c,j} \frac{P_{X^{n}|W_{1}^{n}}(x^{n}|W_{1}(c,j))}{2^{nR_{j}}} - Q_{X}^{n} \right| \right|_{\text{TV}}. \label{eq:TVD_output_noisy}
\end{equation}
The probability distribution of $C$ is moved outside the TV distance expression since the common randomness is independent of $X^{n}$ by the problem definition. By the soft covering lemma \cite[Lemma~IV.1]{CuffChannelSynthesis} the two TV distances in (\ref{eq:TVD_output_noisy}) and (\ref{eq:TVD_input_noisy}) goes to zero as $n \rightarrow \infty$ if we impose the following rates
\begin{align}
    &R_{J} + R_{0} > I(W_{1};Y), \label{eq:y_rate_soft_cov} \\
    &R_{J} > I(W_{1};X). \label{eq:x_rate_soft_cov}
\end{align}
We can set the rate of achievable communication between the encoder and the decoder to $R_{J} = I(W_{2};\widetilde{Y})$, which is achievable since we have
\begin{align}
    nI(W_{2};\widetilde{Y})= I(W_{2}^{n};\widetilde{Y}^{n}) \stackrel{(a)}{\leq} I(\widetilde{X}^{n};\widetilde{Y}^{n}) = nI(\widetilde{X};\widetilde{Y}) \label{eq:rel_com_case2}
\end{align}
where (a) follows by the data processing inequality in the Markov chain $W_{2}^{n} - \widetilde{X}^{n} - (\widetilde{Y}^{n},\widetilde{Z}^{n})$. By the setting $R_{J}=I(W_{2};\widetilde{Y})$, (\ref{eq:x_rate_soft_cov}) is satisfied by (\ref{eq:jscc_case2}) and (\ref{eq:y_rate_soft_cov}) is satisfies by (\ref{eq:proto_coding_ineq}). Moreover, by the typical average lemma \cite[pp. 26]{Elgamalbook} the distortion constraint in (\ref{eq:dist_SRDP}) is satisfied since all tuples $(x^{n},w^{n},y^{n})$ are in the jointly typical set with high probability as $n \rightarrow \infty$.

Since the TV distance in (\ref{eq:tvd_noisy}) goes to zero as $n \rightarrow \infty$, by property 2 of \cite[Lemma~3]{OSRBAmin} there exists an instance of $F=f$ such that the TV distance in (\ref{eq:tvd_noisy}) also goes to zero when conditioning on this instance of $F$ and it therefore exists a reliable encoder-decoder pair. Similiarly to the proof of Theorem~\ref{theo:SRDPNoiseless}, we can now identify an encoder-decoder pair based on this instance of $F=f$ that achieves (\ref{eq:real_SRDP}), (\ref{eq:dist_SRDP}), and (\ref{eq:secrecy_noisyRDP}).

The cardinality of $\mathcal{W}_{1}$ can be limited to $|\mathcal{W}_{1}| \leq |\mathcal{X}|^{2} + 1$ and the cardinality of $\mathcal{W}_{2}$ can be limited to $\mathcal{W}_2 \leq |\mathcal{\widetilde{X}}| + 1$ by the support lemma.

We remark that the steps in our proof of Theorem~\ref{theo:SRDPnoisy} is similar to \cite{cervia2020remote}, however our Theorem considers a different secrecy constraint as well as a realism constraint. Furthermore, the binning structure in our proof is different as well.
\end{IEEEproof}

\section{Secure RDP over Noiseless Channels with Side Information}
%
%
\subsection{Proof for Theorem~\ref{theo:bothSI}}\label{app:thmBothSI}
\begin{IEEEproof}[Proof sketch for achievability]
For the achievability proof of Theorem~\ref{theo:bothSI}, we use the OSRB method with step analogous to the proof for Theorem~\ref{theo:SRDPNoiseless} and thus the following subsections are mostly a proof outline that is highlighting the differences in the proofs.

We fix the distribution $P_{XZUY}$ such that $\mathbb{E}[d(X^{n},Y^{n}] \leq D + \epsilon_{n}$ for a $\epsilon_{n} > 0$ that tends to zero as $n \rightarrow \infty$. We further introduce additional common randomness $F \in [1:2^{n\widetilde{R}}]$ that is shared between the encoder and the decoder which, as previously, which represents the choice of codebook for the encoder and the decoder. We generate the auxiliary random variable $U^{n}$ and assign bin indices $C$, $S$, the common randomness and the message index in the model, and $F$ in two ways:

(i) \emph{random binning} where the auxiliary random variable sequence $U^{n}$ is generated in i.i.d according to the distribution $P_{U}^{n}$ factorizes as $P_{X^{n}Z^{n}U^{n}} = Q_{XZ}^{n} P_{U^{n}|X^{n}Z^{n}}$. We assign bin indices $C$, $S$, and $F$ to each sequence $u^{n}$ according to $C = \varphi_{1}(u^{n})$, $S = \varphi_{2}(u^{n})$, and $F = \varphi_{3}(u^{n})$, which maps each sequence $u^{n}$ uniformly at random to an index $C \in [1:2^{nR_{0}}]$, $S \in [1:2^{nR}]$, and $F \in [1:2^{n\widetilde{R}}]$ respectively. This means that the encoder picks the message according to the distribution $P^{\text{RB}}_{S|U^{n}Z^{n}} \sim \text{Unif}[1:2^{nR}]$. The decoder $P^{\text{RB}}_{\hat{U}^{n}|CFSZ^{n}}$ recover $U^{n}$ from $(C,S,F,Z^{n})$ and produces the output $\hat{U}^{n}$ which the produces $Y^{n}$ according to $P^{\text{RB}}_{Y^{n}|\hat{U}^{n}}$.

(ii) \emph{random coding} where the common randomness and the encoder-decoder pair, $C$ and $F$, is generated uniformly at random according to $P_{C}$ and $P_{F}$ respectively. The auxiliary random variable sequence $U^{n}$ is generated according to $P^{\text{RB}}_{U^{n}|CFX^{n}Z^n}$, which is the same as in the random binning and the relationship between this and the random binning distributions is derived in the same way as in the proof for Theorem~\ref{theo:SRDPNoiseless}. The selection of the message $S$ and the decoder is the same as for the random binning, which is $P^{\text{RB}}_{S|U^{n}Z^{n}}$ and $P^{\text{RB}}_{\hat{U}^{n}|CFSZ^{n}}$ with $P^{\text{RB}}_{Y^{n}|\hat{U}^{n}}$.

Similar to the previous proofs, the two ways of assigning the bin indices induce two different probability distributions, denote $P^{\text{RB}}$ and $P^{\text{RC}}$ for the random binning and the random coding respectively
\begin{align}
    P^{\text{RB}} =& Q^{n}_{XZ} P_{U^{n}|X^{n}Z^{n}} P_{C|U^{n}} P_{F|U^{n}} P_{S|U^{n}Z^{n}}^{\text{RB}} P^{\text{RB}}_{\hat{U}^{n}|CFSZ^{n}} \nonumber \\
    &\cdot P^{\text{RB}}_{Y^{n}|\hat{U}^{n}} \nonumber \\ 
    =& P_{X^{n}Z^{n}CF} P_{U^{n}|CFX^{n}Z^{n}}^{\text{RB}} P_{S|U^{n}Z^{n}}^{\text{RB}} P^{\text{RB}}_{\hat{U}^{n}|CFSZ^{n}} \nonumber \\
    &\cdot P^{\text{RB}}_{Y^{n}|\hat{U}^{n}}, \\
    P^{\text{RC}} =& Q_{XZ}^{n} P_{C} P_{F} P_{U^{n}|CFX^{n}Z^{n}}^{\text{RB}} P_{S|U^{n}Z^{n}}^{\text{RB}} P^{\text{RB}}_{\hat{U}^{n}|CFSZ^{n}} \nonumber \\
    &\cdot P^{\text{RB}}_{Y^{n}|\hat{U}^{n}}.
\end{align}

The TV distance between $P^{\text{RB}}$ and $P^{\text{RC}}$ can be bounded in the same way as in equation (\ref{eq:TVDsums}), which results in
\begin{align}
    ||P^{\text{RB}} - P^{\text{RC}}||_{\text{TV}} \leq ||P_{X^{n}Z^{n}CF} - Q_{XZ}^{n} P_{C} P_{F}||_{\text{TV}}. \label{eq:tvdSIatBoth}
\end{align}

As in the previous achievability proofs, the distance in equation (\ref{eq:tvdSIatBoth}) above goes to zero as $n \rightarrow \infty$ if we impose the following rates. The indices $F$ and $C$ and jointly uniformly distributed and independent of $(X^{n},Z^{n})$ if, by \cite[Theorem~1]{OSRBAmin}, we have the rates
\begin{equation}
    \widetilde{R} + R_{0} < H(U|XZ). \label{eq:indepSIatBoth}
\end{equation}

The public indices $S$ and $F$ are independent of $Y^{n}$ and jointly uniformly distributed if we have the following rates by \cite[Theorem~1]{OSRBAmin}
\begin{equation}
    R + \widetilde{R} < H(U|Y). \label{eq:secrecy}
\end{equation}
This ensures that the secrecy constraint (\ref{eq:secrecy_SRDPwithSI}) is satisfied by the same arguments as in the proof of Theorem~\ref{theo:SRDPNoiseless}.

If we fix the decoder $P^{\text{RB}}_{\hat{U}^{n}|CFSZ^{n}}$ as a Slepian-Wolf decoder, it can reliably produce $U^{n}$ from $(C,F,S,Z^{n})$ if we have the following rates  by \cite[Lemma~1]{OSRBAmin}
\begin{equation}
    R + \widetilde{R} + R_{0} > H(U|Z). \label{eq:reconstruction}
\end{equation}

Similar to the proof of Theorem~\ref{theo:SRDPNoiseless}, the induced distribution from the codebook of $U^{n}(c,s)$ specified by $C \in [1:2^{nR_{0}}]$ and $S \in [1:2^{nR}]$ in the system depicted in Fig.~\ref{fig:SecRDPwithSIatBoth}, approaches the desired output distribution and the given input distribution if we have $R > I(U;X|Z)$ and $R + R_{0} > I(U;Y|Z) - H(Z|Y)$ by applying \cite[Corollary~VII.5]{CuffChannelSynthesis} similar to \cite[Section~IV]{WagnerHamidGunduzarXiv2025}.  Applying Fourier-Motzkin elimination \cite[Appendix~D]{Elgamalbook} to the rates in (\ref{eq:indepSIatBoth}) - (\ref{eq:reconstruction}) and the rates given by \cite[Corollary~VII.5]{CuffChannelSynthesis}, we have
\begin{align}
    & R = I(U;X|Z) + \delta, \label{eq:messageRateUnsec} \\
    & R_{0} = I(U;Y) - I(U;Z) + \delta, \label{eq:secCR} \\
    & R + R_{0} = I(U;Y|Z) - H(Z|Y) + \delta
\end{align}
for some $\delta >0$ which goes to zero as $n \rightarrow \infty$.

Moreover, the distortion constraint is fulfilled by the Typical Average lemma \cite[pp.~26]{Elgamalbook} since the tuples $(x^{n},u^{n},y^{n})$ are in the typical set with high probability as $n \rightarrow \infty$. With the same argument as in the proof of Theorem~\ref{theo:SRDPNoiseless}, \cite[2) in Lemma~3]{OSRBAmin} allows us to identify an encoder-decoder conditioned on an instance of $F=f$ that achieves (\ref{eq:real_SRDP}), (\ref{eq:dist_SRDP}), and (\ref{eq:secrecy_SRDPwithSI}) since (\ref{eq:tvdSIatBoth}) goes to zero as $n\rightarrow \infty$.
\end{IEEEproof}

%
%
\subsection{Proof for Theorem~\ref{theo:onlyDec}}\label{app:thmOnlyDec}
\begin{IEEEproof}[Proof sketch]
In the same ways as in the proof of Theorem~\ref{theo:bothSI}, we fix the distribution $P_{XZUY}$ such that $\mathbb{E}[d(X^{n},Y^{n}] \leq D + \epsilon_{n}$ for a $\epsilon_{n} > 0$ that tends to zero as $n \rightarrow \infty$. We introduce additional common randomness $F \in [1:2^{n\widetilde{R}}]$ that is shared between the encoder and the decoder which, as previously, representing the choice of codebook for the encoder and the decoder. We generate the auxiliary random variable $U^{n}$ and assign bin indices for common randomness $C$, and the message index $S$, and the codebook choice $F$ in two ways:

(i) \emph{random binning} where the auxiliary random variable sequence $U^{n}$ is generated in i.i.d according to the distribution $P_{U}^{n}$ factorizes as $P_{X^{n}Z^{n}U^{n}} = Q_{XZ}^{n} P_{U^{n}|X^{n}}$. We assign bin indices $C$, $S$, and $F$ to each sequence $u^{n}$ according to $C = \varphi_{1}(u^{n})$, $S = \varphi_{2}(u^{n})$, and $F = \varphi_{3}(u^{n})$, which maps each sequence $u^{n}$ uniformly at random to an index $C \in [1:2^{nR_{0}}]$, $S \in [1:2^{nR}]$, and $F \in [1:2^{n\widetilde{R}}]$ respectively. The decoder $P^{\text{RB}}_{\hat{U}^{n}|CFSZ^{n}}$ recover $U^{n}$ from $(C,S,F,Z^{n})$ and produces the output $\hat{U}^{n}$ which the produces $Y^{n}$ according to $P^{\text{RB}}_{Y^{n}|\hat{U}^{n}}$.

(ii) \emph{random coding} where the common randomness $C$ and the encoder-decoder pair $F$ is generated uniformly at random according to $P_{C}$ and $P_{F}$ respectively. The auxiliary random variable sequence $U^{n}$ is generated according to $P^{\text{RB}}_{U^{n}|CFX^{n}}$, which is the same as in the random binning and the relationship is derived in the same way as in the proof for Theorem~\ref{theo:SRDPNoiseless}. The selection of the message $S$ and the decoder is the same as for the random binning, which is $P^{\text{RB}}_{S|U^{n}}$ and $P^{\text{RB}}_{\hat{U}^{n}|CFSZ^{n}}$ with $P^{\text{RB}}_{Y^{n}|\hat{U}^{n}}$.

The distributions that are induced in the random binning $P^{\text{RB}}$ and the random coding $P^{\text{RC}}$ are
\begin{align}
    P^{\text{RB}} =& Q_{XZ}^{n} P_{U|X}^{n} P_{C|U^{n}} P_{F|U^{n}} P_{S|U^{n}}^{\text{RB}} P_{\hat{U}^{n}|CFSZ^{n}}^{\text{RB}} \nonumber \\
    &\cdot P_{Y^{n}|\hat{U}^{n}}^{\text{RB}} \nonumber \\
    =& P_{X^{n}Z^{n}CF} P_{U^{n}|CFX^{n}}^{\text{RB}} P_{S|U^{n}}^{\text{RB}} P_{\hat{U}^{n}|CFSZ^{n}}^{\text{RB}} \nonumber \\
    &\cdot P_{Y^{n}|\hat{U}^{n}}^{\text{RB}},  \\
    P^{\text{RC}} = & Q_{XZ}^{n} P_{C} P_{F} P_{U^{n}|CFX^{n}}^{\text{RB}} P_{S|U^{n}}^{\text{RB}} P_{\hat{U}^{n}|CFSZ^{n}}^{\text{RB}} \nonumber \\
    &\cdot P_{Y^{n}|\hat{U}^{n}}^{\text{RB}}.
\end{align}

In the same way as in equation (\ref{eq:TVDsums}), the TV distance between the two distributions can be upper bounded by
\begin{align}
    ||P^{\text{RB}} - P^{\text{RC}}||_{\text{TV}} \leq ||P_{X^{n}Z^{n}CF} - Q_{XZ}^{n} P_{C} P_{F}||_{\text{TV}}. \label{eq:tvdSIatDec}
\end{align}

Following similar reasoning as in the proof of Theorem~\ref{theo:SRDPNoiseless}, the TV distance in (\ref{eq:tvdSIatDec}) tend to zero when $n \rightarrow \infty$ if we impose the following rate constraints. The index $C$ and the public index $F$ are almost independent of $(X^{n},Z^{n})$ and jointly uniformly distributed if we impose \cite[Theorem~1]{OSRBAmin}
\begin{equation}
    \widetilde{R} + R_{0} < H(U|XZ). \label{eq:indepSIatDec}
\end{equation}
The public indices $S$ and $F$ are almost independent of $Y^{n}$, satisfying the secrecy constraint (\ref{eq:secrecy_SRDPwithSI}), and jointly uniformly distributed if we have \cite[Theorem~1]{OSRBAmin}
\begin{equation}
    R + \widetilde{R} < H(U|Y), \label{eq:secrecySIatDec}
\end{equation}
which ensures that the secrecy constraint (\ref{eq:secrecy_SRDPwithSI}) is satisfied following the same arguments as in the proof of Theorem~\ref{theo:SRDPNoiseless}. Lastly, fix the decoder $P^{\text{RB}}_{\hat{U}^{n}|CFSZ^{n}}$ to a Slepian-Wolf (SW) decoder $P^{\text{SW}}_{\hat{U}^{n}|CFSZ^{n}}$ which can reliably estimate $U^{n}$ from $(C,S,F,Z^{n})$ such that the sequence estimation error goes to zero when $n \rightarrow \infty$ if we impose \cite[Lemma~1]{OSRBAmin}
\begin{equation}
    R + \widetilde{R} + R_{0} > H(U|Z). \label{eq:reconstructionSIatDec}
\end{equation}
Applying Fourier-Motzkin elimination to the rates in (\ref{eq:indepSIatDec}) - (\ref{eq:reconstructionSIatDec}), we have
\begin{align}
    &R = I(U;X|Z) + \delta, \\
    &R_{0} = I(U;Y) - I(U;Z) + \delta
\end{align}
where $\delta \rightarrow 0$ when $n \rightarrow \infty$. We have
\begin{align}
    I(U;X|Z) =& H(U|Z) - H(U|ZX) \nonumber \\
    =& H(U|Z) - H(U|X) \nonumber \\
    =& I(U;X) - I(U;Z)
\end{align}
since $Z - X - U$ forms a Markov chain by the generation of the auxiliary random variable in the random binning.

Moreover, applying \cite[Corollary~VII.6]{CuffChannelSynthesis} in the system in Fig.~\ref{fig:SecRDPwithSIatDec} yields the following sum-rate for $R+ R_{0}$ for achieving the realism constraint (\ref{eq:real_SRDP}) using the side information $Z^{n}$ with local channel synthesis in the decoder
\begin{equation}
    R + R_{0} \geq I(U;Y|Z).
\end{equation}

Furthermore, the distortion constraint is fulfilled by the typical average lemma \cite[pp.~26]{Elgamalbook} due to the tuples $(x^{n},u^{n},y^{n})$ being in the typical set with high probability as $n \rightarrow \infty$. With the same argument as in the proof of Theorem~\ref{theo:SRDPNoiseless}, \cite[2) in Lemma~3]{OSRBAmin} allows us to identify an encoder-decoder conditioned on this instance of $F=f$ that achieves (\ref{eq:real_SRDP}), (\ref{eq:dist_SRDP}), and (\ref{eq:secrecy_SRDPwithSI}) since (\ref{eq:tvdSIatDec}) goes to zero as $n\rightarrow \infty$. Finally, the cardinality bound is shown using the support lemma.
\end{IEEEproof}

\section{Examples}
\subsection{Proof for Corollary~\ref{cor:simpleEx}}\label{app:corSimpleEx}
\begin{IEEEproof}[Proof sketch for achievability]
Consider a BSC from $X$ to $U$ with cross-over probability $\alpha$ and a BSC from $U$ to $Y$ with cross-over probability $\beta$, such that the conditions in (\ref{eq:real_SRDP})-(\ref{eq:secrecy_SRDP}) and the Markov chain $X-U-Y$ are satisfied. With the given distributions for the random variables in the system model, the rates in Theorem~\ref{theo:SRDPNoiseless} evaluate to
\begin{align}
    R \geq I(U;X) = 1 - H_{b}(\alpha), \\
    R_{0} \geq I(U;Y) = 1 - H_{b}(\beta)
\end{align}
Moreover, with Hamming distortion, we have
\begin{align}
    D \geq \mathbb{E}[d(X,Y)] =& \sum_{x,y\in\{0,1\}} p_{XY}(x,y) \mathbb{1}\{x \neq \!y\} \nonumber \\
    =& \alpha * \beta
\end{align}
where $\mathbb{1}\{\cdot\}$ is the indicator function.
\end{IEEEproof}

\subsection{Proof for Proposition~\ref{prop:GaussianSingleLetterFamily}}\label{app:propGaussian}
\begin{IEEEproof}[Proof Sketch]
Since $U=X+N_{U}$ with $N_{U}$ independent of $(X,Z)$, the joint law of $(X,Z,U)$ is Gaussian, and we have $P_{ZU|X} = P_{Z|X}P_{U|X}$, which proves \eqref{eq:gaussMC1}. Define
\begin{align}
&V \triangleq
\begin{bmatrix}
U\\
Z
\end{bmatrix}, \\
&K_{X,[U\,Z]}
\triangleq
\begin{bmatrix}
\mathrm{Cov}(X,U) & \mathrm{Cov}(X,Z)
\end{bmatrix}, \\
&K_{[U\,Z]}
\triangleq
\begin{bmatrix}
\mathrm{Var}(U) & \mathrm{Cov}(U,Z)\\
\mathrm{Cov}(Z,U) & \mathrm{Var}(Z)
\end{bmatrix}.
\end{align}
Since $(X,U,Z)$ is jointly Gaussian, the conditional mean of $X$ given $(U,Z)$ is affine in $(U,Z)$. Standard Gaussian conditioning formulas then give
\begin{align}
    &\mathbb{E}[X|U,Z]
    = K_{X,[U\,Z]} K_{[U\,Z]}^{-1} V \nonumber\\
    &=
    \begin{bmatrix}
        \mathrm{Cov}(X,U) & \mathrm{Cov}(X,Z)
    \end{bmatrix} \nonumber \\
    & \qquad\times
    \begin{bmatrix}
        \mathrm{Var}(U) & \mathrm{Cov}(U,Z)\\
        \mathrm{Cov}(Z,U) & \mathrm{Var}(Z)
    \end{bmatrix}^{-1}
    \begin{bmatrix}
        U \\ Z
    \end{bmatrix} \nonumber\\
    &=
    \begin{bmatrix}
        1 & \eta
    \end{bmatrix}
    \begin{bmatrix}
        1+s(\nu) & \eta\\
        \eta & 1
    \end{bmatrix}^{-1}
    \begin{bmatrix}
        U \\ Z
    \end{bmatrix} \nonumber \\ &=
    \frac{1-\eta^{2}}{1+s(\nu)-\eta^{2}}\,U
    + \frac{\eta\,s(\nu)}{1+s(\nu)-\eta^{2}}\,Z
\end{align}

which proves \eqref{eq:SLinearForm}. Moreover, from the law of total variance, we have
\begin{align}
    \mathbb{E}[\mathrm{Var}(X|U,Z)] =& \mathrm{Var}(X|U,Z)
    = 1 - \mathrm{Var}\big(\mathbb{E}[X|U,Z]\big) \nonumber \\
    =& \frac{s(\nu)(1-\eta^{2})}{1+s(\nu)-\eta^{2}}
    \overset{(a)}{=} 1-\nu
\end{align}
where $(a)$ follows by substituting \eqref{eq:sOfNu}. Hence, we have $\mathrm{Var}(S)=\nu$.

Next, \eqref{eq:gaussYConstruction}--\eqref{eq:gaussYNnoise} show that $Y$ is generated from $(U,Z)$ through $S=\mathbb{E}[X|U,Z]$ and an independent Gaussian $N_{Y}$. Therefore, we have $P_{XUZY}=P_{XUZ}P_{Y|UZ}$, which proves \eqref{eq:gaussMC2}. Also, $Y$ is Gaussian with zero mean and variance
$\displaystyle \mathrm{Var}(Y) = \frac{\rho_{\Delta}^{2}}{\nu^{2}}\mathrm{Var}(S)+ \mathrm{Var}(N_{Y}) = \frac{\rho_{\Delta}^{2}}{\nu^{2}}\nu + 1-\frac{\rho_{\Delta}^{2}}{\nu} = 1$, which proves \eqref{eq:gaussRealism}.

Furthermore, we have
$\mathbb{E}[XS] = \mathbb{E}\!\big[\mathbb{E}[XS|U,Z]\big]= \mathbb{E}\!\big[S\,\mathbb{E}[X|U,Z]\big]= \mathbb{E}[S^{2}]
= \mathrm{Var}(S)= \nu$. Hence, we have $\mathbb{E}[XY]
= \frac{\rho_{\Delta}}{\nu}\mathbb{E}[XS]
= \rho_{\Delta}$, and substituting this into \eqref{eq:gaussDistIdentity} proves \eqref{eq:gaussExactDist}.


We next evaluate the three single-letter quantities. By \eqref{eq:gaussMC1}, we have $I(U;X)-I(U;Z)=I(U;X|Z)$. Since $(X,U,Z)$ is jointly Gaussian, we have
\begin{align}
    I(U;X|Z)
    =& \frac{1}{2}\log\frac{\mathrm{Var}(X|Z)}{\mathrm{Var}(X|U,Z)} \nonumber \\
    =& \frac{1}{2}\log\frac{1-\eta^{2}}{1-\nu},
\end{align}
which proves \eqref{eq:R1Gaussian}.

To evaluate \eqref{eq:R3Gaussian}, note that $S$ is a deterministic function of $(U,Z)$ and $Y$ depends on $(U,Z)$ only through $S$. Therefore, we have $I(U;Y|Z)=I(S;Y|Z)$. Also, we obtain
\begin{align}
    \mathrm{Cov}(S,Z) =& \mathbb{E}[SZ] =
    \mathbb{E}\big[Z\,\mathbb{E}[X|U,Z]\big] \nonumber \\
    =& \mathbb{E}\big[\mathbb{E}[ZX|U,Z]\big] =\mathbb{E}[ZX]
    = \eta,
\end{align}
so using the standard scalar Gaussian conditional-variance formula, we have
\begin{equation}
    \mathrm{Var}(S|Z)=\mathrm{Var}(S)-\frac{\mathrm{Cov}(S,Z)^{2}}{\mathrm{Var}(Z)}=\nu-\eta^{2}.
\end{equation}
Conditioned on $Z$, the channel from $S$ to $Y$ is Gaussian and we have
\begin{equation}
    Y=\frac{\rho_{\Delta}}{\nu}S+N_{Y},
\end{equation}
with noise variance $1-\rho_{\Delta}^{2}/\nu$. Therefore,
\begin{align}
    I(U;Y|Z)
    &= \frac{1}{2}\log\!\left(
    1+\frac{\frac{\rho_{\Delta}^{2}}{\nu^{2}}(\nu-\eta^{2})}
    {1-\rho_{\Delta}^{2}/\nu}\right) \nonumber \\
    &= \frac{1}{2}\log\frac{\nu^{2}-\eta^{2}\rho_{\Delta}^{2}}
    {\nu(\nu-\rho_{\Delta}^{2})},
\end{align}
which proves \eqref{eq:R3Gaussian}.

Finally, $(U,Y)$ and $(U,Z)$ are both jointly Gaussian. Since we have
\begin{equation}
    \mathrm{Var}(U)=1+s(\nu), \qquad \mathrm{Cov}(U,Z)=\eta,
\end{equation}
we obtain
\begin{equation}
    I(U;Z)=\frac{1}{2}\log\frac{1+s(\nu)}{1+s(\nu)-\eta^{2}}.\label{eq:GaussianIUZ}
\end{equation}
Moreover, consider
\begin{equation}
    \mathbb{E}[US]
    = \mathbb{E}\!\big[U\,\mathbb{E}[X|U,Z]\big]
    = \mathbb{E}[UX]
    = 1,
\end{equation}
so we have
\begin{equation}
    \mathrm{Cov}(U,Y)=\frac{\rho_{\Delta}}{\nu}.
\end{equation}
Hence, we obtain
\begin{equation}
    I(U;Y)
    = \frac{1}{2}\log\frac{1+s(\nu)}{1+s(\nu)-\rho_{\Delta}^{2}/\nu^{2}}.\label{eq:GaussianIUY}
\end{equation}
Subtracting (\ref{eq:GaussianIUZ}) and (\ref{eq:GaussianIUY}) proves \eqref{eq:R2Gaussian}. The region \eqref{eq:GaussFamilyRegion} then follows immediately.
\end{IEEEproof}

\end{appendices}

\bibliographystyle{IEEEtran}
\bibliography{IEEEabrv,bibliography}

\end{document}